\documentclass[journal]{IEEEtran} %
\usepackage{cite}
\usepackage{textcomp}

\usepackage{graphics} %
\usepackage{graphicx}

\usepackage{algorithmic}
\usepackage{times} %
\usepackage{amsmath} %
\usepackage{amssymb}  %

\usepackage{tikz} %
\usepackage{pgfplots} %

\usepackage[ruled,vlined,linesnumbered]{algorithm2e}
\usepackage{epstopdf}
\usepackage{cases}

\usepackage{color}

\usepackage{lipsum}%
\usepackage{verbatim}

\usepackage{amsthm}
\usepackage{siunitx}

\usepackage{mathtools}
\usepackage{aligned-overset}

\usepackage{circuitikz}
\usetikzlibrary{positioning}
\usetikzlibrary{backgrounds}
\usetikzlibrary{shapes,arrows,chains}

\newtheorem{defi}{Definition}
\newtheorem{theo}{Theorem}
\newtheorem{rema}{Remark}

\newtheorem{asum}{Assumption}
\newtheorem{prop}{Proposition}

\newcommand{\svar}{t}
\newcommand{\tvar}{s}

\newcommand{\auxvar}{\tau_{\max}}

\newcommand{\lvar}{\Lambda}
\newcommand{\tpar}{\text{\normalfont par}}
\newcommand{\refer}{\text{\normalfont ref}}
\newcommand{\barj}{{\tilde j}}
\newcommand{\maxvar}{k_{\max}}

\newcommand{\newi}{{i,\varl}}
\newcommand{\newij}{{{i_j},\varl}}
\newcommand{\varl}{l}
\newcommand{\tc}{c}

\newcommand{\windowC}{C(x(\tvar_j),\eta(\tvar_j),c)}%

\newcommand{\itilde}{1}
\newcommand{\itildel}{{1,l}}
\newcommand{\ivar}[1]{{i_{#1}}}

\newcommand{\maxvart}[1]{k_{2}}

\newcommand{\norm}[1]{\left\lVert#1\right\rVert}
\newcommand{\abs}[1]{\left|#1\right|}

\colorlet{istorange}{orange}
\colorlet{istgreen}{green!50!black}
\colorlet{istblue}{blue} %
\colorlet{istred}{red!90!black}

\begin{document}
	\title{Robust dynamic  self-triggered control for nonlinear systems using hybrid Lyapunov functions
	 }
	\author{ Michael Hertneck and Frank Allg\"ower
		\thanks{Funded by Deutsche
			Forschungsgemeinschaft (DFG, German Research Foundation) under Germany’s
			Excellence Strategy - EXC 2075 - 390740016 and under grant
			AL 316/13-2 - 285825138. We acknowledge the support by the Stuttgart
			Center for Simulation Science (SimTech).}
		\thanks{The authors are with the University of Stuttgart, Institute for Systems Theory and Automatic Control, 70569 Stuttgart, Germany (email:
			{\{{hertneck}{,allgower}\}}{@ist.uni-stuttgart.de}).}}
	
	\maketitle
	
	\begin{abstract}
		Self-triggered control (STC) is a resource efficient approach to determine sampling instants for Networked Control Systems. At each sampling instant, an STC mechanism determines not only the control inputs but also the next sampling instant. 
		In this article, an STC framework for perturbed nonlinear systems is proposed. In the framework, a dynamic variable is used in addition to current state information to determine the next sampling instant, rendering the STC mechanism dynamic. Using dynamic variables has proven to be powerful for increasing sampling intervals for the closely related concept of event-triggered control, but has so far not been exploited for STC.
		Two variants of the dynamic STC framework are presented. The first variant can be used without further knowledge on the disturbance and leads to guarantees on input-to-state stability. The second variant exploits a known disturbance bound to determine sampling instants and guarantees asymptotic stability of a set containing the origin.  In both cases, hybrid Lyapunov function techniques are used to derive the respective stability guarantees.
		Different choices for the dynamics of the dynamic variable, that lead to different particular STC mechanisms, are presented for both variants of the framework.  The resulting dynamic STC mechanisms are illustrated with two numerical examples to emphasize their benefits in comparison to existing static STC approaches.  
	\end{abstract}
	

\section{Introduction}
Many modern control applications, e.g., in the field of Networked Control Systems (NCS), involve the implementation of feedback laws on shared hardware or with limited communication resources \cite{hespanha2007survey}. 
Classically, feedback laws are evaluated periodically with a sampling frequency that is determined before the runtime of the controller. 
However, such periodic sampling may lead in many situations to a waste of resources, as reported, e.g., in \cite{astroem1999comparison,arzen1999simple}. Therefore, event- and self-triggered control have been developed as alternatives to periodic sampling \cite{heemels2012introduction}. 
In event-triggered control (ETC), a state-dependent trigger condition is used to determine sampling instants. The trigger condition is monitored continuously or at periodic time instants and a transmission of sampled state information is triggered if the trigger condition is fulfilled. 
While ETC can typically significantly reduce the required amount of transmissions, the continuous/periodic monitoring of the trigger condition may be impractical in many NCS setups, as it requires that the trigger mechanism has access to current state information at each time the trigger condition shall be checked. Moreover, the resulting traffic patterns for ETC are hard to predict, which may result in undesired effects like packet loss or highly time varying delays~\cite{blind2011analysis}.

In self-triggered control (STC), the controller determines at each sampling instant based on sampled state information when the next sample should be taken, thus lowering the effort for monitoring the plant state. STC can reduce the network load for NCS in comparison to periodic sampling significantly, as it has been demonstrated in \cite{mazo2009self,anta2010sample}. Moreover, it allows for efficient scheduling of sampling instants \cite{samii2010dynamic}. For linear systems, a variety of STC approaches are available, see, e.g., \cite{heemels2012introduction,brunner2019event} and the references therein. For nonlinear systems, there are only a few approaches available to to this date. 
In \cite{anta2010sample,delimpaltadakis2020isochronous,delimpaltadakis2020region} a conservative prediction for the trigger condition of ETC is used to determine sampling instants. For an efficient online evaluation, the state-space is divided into regions with the same sampling interval.    
In \cite{benedetto2011digital,tiberi2013simple,liu2015small,theodosis2018self,proskurnikov2019lyapunov}, Lipschitz continuity properties are used to determine sampling instants such that the decrease of a Lyapunov function can be guaranteed.  

In all aforementioned works on nonlinear STC, only the current state of the system is taken into account to determine sampling instants. However, for the related approach of STC, it has been demonstrated, that taking also information about the past system behavior into account for the trigger condition can reduce the required amount of sampling significantly \cite{girard2015dynamic}. 
In \cite{girard2015dynamic} the past system behavior is encoded in a dynamic variable, such that an averaged version of a purely state-dependent trigger rule can be used to determine sampling instants, leading to significantly increased sampling intervals.
Similar benefits as for such \textit{dynamic} ETC can also be expected for STC if the past system behavior is taken into account, yet no dynamic STC approaches exist to this date. 

In a preliminary study for the work at hand \cite{hertneck21dynamic}, we have presented a \textit{dynamic} STC mechanism based on hybrid Lyapunov functions that bridges this gap. In \cite{hertneck21dynamic}, a dynamic variable is used to encode the past system behavior. In particular, therein the dynamic variable is chosen as a filtered average of the values of a Lyapunov function at previous sampling instants. At each sampling instant, the next sampling instant is selected such that the value of the Lyapunov function at the next sampling instant is bounded by the value of the dynamic variable. Hybrid Lyapunov techniques adapted from \cite{carnevale2007lyapunov,nesic2009explicit,hertneck20stability} are used in \cite{hertneck21dynamic} to guarantee this bound and to derive stability guarantees for the dynamic STC mechanism. It is demonstrated in \cite{hertneck21dynamic} that the proposed dynamic STC mechanism based on hybrid Lyapunov functions can significantly reduce the required amount of samples in comparison to static STC.    

In this article, we extend the results from \cite{hertneck21dynamic} in various directions. The main contributions compared to the preliminary results in \cite{hertneck21dynamic} are: 

(1) We present a general framework for dynamic STC based on hybrid Lyapunov functions that includes different particular dynamic STC mechanisms. While in \cite{hertneck21dynamic}, the dynamics of the dynamic variable were fixed to be a finite impulse response (FIR) filter for the Lyapunov function, the general framework in this article allows different choices. We present two particular alternatives, namely choosing the dynamic variable as an infinite impulse response (IIR) filter for the Lyapunov function and choosing it such that it implements a tunable reference function as a bound for the Lyapunov function and illustrate these options with numerical examples. Further choices can be included in the general framework as well. 

(2) We extend the results from \cite{hertneck21dynamic} to perturbed nonlinear systems. Note that only few results on  STC for perturbed nonlinear systems are available, and all require knowledge of  bound on the disturbance (see \cite{delimpaltadakis2020region} and the references therein for a literature overview). We present two different variants of the dynamic STC framework such that either input-to-state stability (ISS) or robust asymptotic stability (RAS) of a sublevel set of the Lyapunov function can be guaranteed. The variant with ISS guarantees requires no prior knowledge on the disturbance signal such as a disturbance bound. The variant that ensures RAS of a level set of the Lyapunov function does require knowledge of a bound on the disturbance, but offers two advantages in comparison to the ISS variant: Firstly, it can be used locally, i.e., assumptions on the system dynamics that are required for using hybrid Lyapunov functions only need to hold locally. Secondly, by taking the disturbance bound explicitly into account for determining sampling intervals, less frequent triggering may be required in some situations in comparison to the ISS variant. 

(3) For the RAS variant of our framework, we modify the dynamic STC mechanism from \cite{hertneck21dynamic} such that it can use different parameters for different level sets of the Lyapunov function. This can significantly reduce the required amount of sampling instants, since the employed techniques for hybrid Lyapunov functions may allow for larger sampling intervals when it can be ensured that the system state stays in certain level sets of the Lyapunov function. 

(4) We present more compact proofs of the main results compared to \cite{hertneck21dynamic}.   

The remainder of this article is organized as follows. The problem setup is described in Section~\ref{sec_setup}. In Section~\ref{sec_concept}, the variant of the general framework for dynamic STC based on hybrid Lyapunov functions that requires no disturbance bound is presented. Specific dynamic STC mechanisms with ISS guarantees for this variant are introduced in Section~\ref{sec_spec}. The variant of the framework that takes explicitly into account a disturbance bound and respective specific dynamic STC mechanisms with RAS guarantees are presented in Section~\ref{sec_loc}. In Section~\ref{sec_ex}, the proposed mechanisms are illustrated with two numerical examples. Section~\ref{sec_conc} concludes the article.

\subsubsection*{Notation and definitions}
The nonnegative real numbers are denoted by  $\mathbb{R}_{\geq 0} $. The natural numbers are denoted by $\mathbb{N}$, and we define $\mathbb{N}_0:=\mathbb{N}\cup  \left\lbrace 0 \right\rbrace $. 
We denote the Euclidean norm by $\abs{\cdot}$ and the infinity norm by $\abs{\cdot}_\infty$. A continuous function $\alpha: \mathbb{R}_{\geq 0} \rightarrow \mathbb{R}_{\geq 0}$ is a class $ \mathcal{K}$ function if it is strictly increasing and $\alpha(0) = 0$. It is a class $\mathcal{K}_\infty$ function if it is a class $\mathcal{K}$ function and it is unbounded. A continuous function $\beta:\mathbb{R}_{\geq 0}\times \mathbb{R}_{\geq 0} \rightarrow \mathbb{R}_{\geq 0}$ is a class $\mathcal{K}\mathcal{L}$ function, if $\beta(\cdot,t)$ is a class $\mathcal{K}$ function for each $t \in\mathbb{R}_{\geq 0}$ and $\beta(q,\cdot)$ is nonincreasing and satisfies $\lim\limits_{t \rightarrow \infty} \beta(q,t) = 0$ for each $q\in\mathbb{R}_{\geq 0}$. A function $\beta:\mathbb{R}_{\geq 0}\times \mathbb{R}_{\geq 0} \times \mathbb{R}_{\geq 0} \rightarrow \mathbb{R}_{\geq 0}$ is a class $\mathcal{K}\mathcal{L}\mathcal{L}$ function if for each $r \geq 0$, $\beta(\cdot,r,\cdot)$ and $\beta(\cdot,\cdot,r)$ are class $\mathcal{K}\mathcal{L}$ functions.

To characterize a hybrid model of the considered NCS, we use the following definitions, that are taken from \cite{goebel2006solutions,heemels2010networked}.

\begin{defi}
	\cite{heemels2010networked} A \textit{compact hybrid time domain} is a set $\mathcal{D} = \bigcup^{J-1}_{j = 0} \left(\left[t_j,t_{j+1}\right],j\right) \subset \mathbb{R}_{\geq0} \times \mathbb{N}_{0}$ with $J \in \mathbb{N}_{0}$ and $0 = t_0 \leq t_1 \dots \leq t_J$. A \textit{hybrid time domain} is a set $\mathcal{D} \subset \mathbb{R}_{\geq0} \times \mathbb{N}_{ 0}$ such that $\mathcal{D} \cap \left(\left[0,T\right] \times \left\lbrace 0,\dots,J\right\rbrace\right)$ is a compact hybrid time domain for each $\left(T,J\right) \in \mathcal{D}$.
\end{defi}
\begin{defi}
	\cite{heemels2010networked} A \textit{hybrid trajectory} is a pair (dom~$\xi$, $\xi$) consisting of the hybrid time domain dom~$\xi$ and a function $\xi$ defined on dom~$\xi$ that is absolutely continuous in $t$ on (dom~$\xi$)$\cap\left(\mathbb{R}_{\geq 0} \times \left\lbrace j \right\rbrace \right)$ for each $j\in\mathbb{N}_{ 0}$. 
\end{defi}
\begin{defi}
	\cite{heemels2010networked}	For the hybrid system $\mathcal{H}$ given by the state space $\mathbb{R}^{n_\xi}$, the input space $\mathbb{R}^{n_w}$ and the data $\left(F,G,C,D\right),$ where $F:\mathbb{R}^{n_\xi}\times\mathbb{R}^{n_w}\rightarrow \mathbb{R}^{n_\xi}$ is continuous, $G:\mathbb{R}^{n_\xi}\rightarrow \mathbb{R}^{n_\xi}$ is locally bounded, and $C$ and $D$ are subsets of $\mathbb{R}^{n_\xi}$, a hybrid trajectory (dom $\xi$,$\xi$) with $\xi:$ dom~$\xi \rightarrow \mathbb{R}^{n_\xi}$ is a \textit{solution to} $\mathcal{H}$ for a locally integrable input function $w:\mathbb{R}_{\geq 0} \rightarrow \mathbb{R}^{n_w}$ if
	\begin{enumerate}
		\item for all $j \in \mathbb{N}_{ 0}$ and for almost all $t\in I_j \coloneqq\left\lbrace t|(t,j) \in \text{dom}~\xi \right\rbrace$  
		we have $\xi(t,j) \in C$ and $\dot{\xi}(t,j) = F(\xi(t,j),w(t))$;
		\item for all $(t,j)\in$ dom~$\xi$ such that $(t,j+1) \in $ dom~$\xi$, we have $\xi(t,j)\in D$ and $\xi(t,j+1) = G(\xi(t,j))$.
	\end{enumerate}
\end{defi}
We thus consider hybrid system models of the form
\begin{equation}
	\begin{alignedat}{2}
		\dot{\xi}(t,j) =& F(\xi(t,j),w(t))\qquad\xi(t,j)&&\in C,\\
		\xi(\svar_{j+1},j+1) =& G(\xi(\svar_{j+1},j))\qquad\xi(\svar_{j+1},j) &&\in D.
	\end{alignedat}
\end{equation}

We sometimes omit the time arguments and write 
\begin{equation}
\label{eq_sys_hyb}
\begin{alignedat}{2}
\dot{\xi} =& F(\xi,w)\quad&&\xi \in C,\\
\xi^+ =& G(\xi)\quad&&\xi \in D,
\end{alignedat}
\end{equation}
where we denoted $\xi(\svar_{j+1},j+1)$ as $\xi^+$. For further details on these definitions, see \cite{goebel2006solutions}. 
\section{Setup}
\label{sec_setup}
In this section, we model first the considered dynamic STC setup as a hybrid system. Then, we present stability definitions for the considered setup and a precise problem statement.
\subsection{Hybrid system model for dynamic STC}
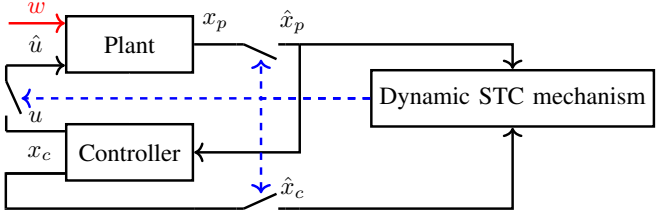
\begin{figure}[tb]
\centering
\resizebox{\columnwidth}{!}{
	\begin{circuitikz}
		[line width=1.0]
		\usetikzlibrary{calc}
		\ctikzset{bipoles/thickness=1}
		\node[draw,minimum width=1.8cm,minimum height=.8cm,anchor=south west] at (0,0) (node1){Plant};
		\node[draw,below= 7mm of node1, minimum width=1.8cm,minimum height=.8cm] (node2) {Controller};
		\draw [color= white,line width = 0pt] (node1) -- node [midway,align = center](helpnode2) {} (node2);
		\node[draw,right= 33mm of helpnode2, minimum width=1.8cm,minimum height=.8cm] (node3) {Dynamic STC mechanism};

		\coordinate[left = 0 mm of node1](Q0a);
		\coordinate[below = 3 mm of Q0a](Q0b);
		\coordinate[left = 0 mm of node2](Q1a);
		\coordinate[above = 3 mm of Q1a](Q1b);
		\coordinate[below = 3 mm of Q1a](Q1c);
		\coordinate[left = 8.4 mm of Q1c](Q1d);
		\coordinate[below = 5 mm of Q1d](Q1e);
		
		\node[left = 7 mm of Q0b] (helpnode0) {};
		\node[left = 7 mm of Q1b] (helpnode1) {};

		\coordinate (Q0) at (helpnode0);
		\coordinate (Q1) at (helpnode1);
		\coordinate (Q2) at (helpnode2);

		\draw [color= white,line width = 0pt] (Q0) -- node [midway,align = center](helpnode3) {} (Q1);

		\coordinate[right = 1 mm of helpnode3](Q3);
		\coordinate[right = 0mm of node1] (Q4); 
		\coordinate[right = 0mm of node2] (Q5a);
		\coordinate[below = 8mm of Q5a] (Q5);
		\coordinate[right = 6mm of node1] (Q6); 
		\coordinate[right = 6mm of Q5] (Q7);
		\coordinate[right = 13mm of Q4] (Q8); 
		\coordinate[right = 13mm of Q5] (Q9);
		\coordinate[above = 0mm of Q5a] (Q10);
		\coordinate[right = 2mm of Q8] (Q11); 
		
		\draw [color= white,line width = 0pt] (Q6) -- node [midway,align = center](helpnode12) {} (Q8);
		\draw [color= white,line width = 0pt] (Q7) -- node [midway,align = center](helpnode13) {} (Q9);
		
		\coordinate[below = 1 mm of helpnode12](Q12);
		\coordinate[above= 1 mm of helpnode13](Q13);
		\coordinate[right = 2mm of Q9] (Q14); 
		
		\coordinate[above = 3mm of Q0a] (Q15); 
		\coordinate[left = 8mm of Q15] (Q16); 
		
		\draw (Q0) to[normal open switch] (Q1);
		\draw (Q6) to[normal open switch,mirror] (Q8);
		\draw (Q7) to[normal open switch] (Q9);
		
 		\draw [-,line width = 1pt] (Q1) -- node [,midway,above = 0.0em,align = center] {${u}$}(Q1b);
 		\draw [->,line width = 1pt] (Q0) -- node [,midway,above = 0.15em,align = center] {$\hat{u}$} (Q0b);
 		\draw [->,line width = 1pt,dashed,color = blue] (node3) -- node [dashed,midway,above = 0.25em] {} (Q3);
 		\draw [-,line width = 1pt] (Q4) -- node [,midway,above = 0.15em,align = center] {${x_p}$}(Q6);
 		\draw [-,line width = 1pt] (Q5) -- node [,midway,above = 0.05em,align = center] {}(Q7);
 		\draw [->,line width = 1pt] (Q11) |- node [,midway,above = 0.05em,align = center] {}(Q10);
		\draw [-,line width = 1pt] (Q8) -- node [,midway,above = 0.05em,align = center] {$\hat{x}_p$}(Q11);
		\draw [->,line width = 1pt] (Q11) -| node [,midway,above = 0.05em,align = center] {}(node3);
		\draw [->,line width = 1pt,dashed,color = blue] (node3) -| node [dashed,midway,above = 0.25em] {} (Q12);
		\draw [->,line width = 1pt,dashed,color = blue] (node3) -| node [dashed,midway,above = 0.25em] {} (Q13);
		\draw [-,line width = 1pt] (Q9) -- node [,midway,above = 0.05em,align = center] {$\hat{x}_c$}(Q14);
		\draw [->,line width = 1pt] (Q14) -| node [,midway,above = 0.05em,align = center] {}(node3);
		\draw [<-,line width = 1pt, color=red] (Q15) -- node [,midway,above = 0.05em,align = center]{$w$}(Q16);
		\draw [-,line width = 1pt] (Q1c) -| node [,midway,above = 0.15em,align = center] {}(Q1e);
		\draw [-,line width = 1pt] (Q1e) |- node [,midway,above = 0.15em,align = center] {}(Q5);
		\draw [-,line width = 1pt] (Q1c) -- node [,midway,above = 0.15em,align = center] {$~x_c$}(Q1d);
	\end{circuitikz}}
\caption{Sketch of the considered setup.}
\label{fig_setup}
\end{figure}
We consider a nonlinear plant that exchanges information with a (possibly) dynamic controller only at discrete sampling
 instants as depicted in Figure~\ref{fig_setup}. The plant is given by
\begin{equation}
\label{eq_plant}
\dot{x}_p = f_p(x_p,\hat{u},w),
\end{equation}
where $x_p(t) \in \mathbb{R}^{n_{x_p}}$ is the plant state with initial condition  $x_p(0) = x_{p,0}$, $\hat{u}(t) \in \mathbb{R}^{n_u}$ is the last input that has been received by the plant and $w(t)\in\mathbb{R}^{n_w}$ is a disturbance input. The controller is given by
\begin{equation}
\label{eq_controller}
\begin{split}
\dot{x}_c = f_c(x_c,\hat{x}_p)\\
u = g_c(x_c,x_p),
\end{split}
\end{equation}
where $x_c(t)\in\mathbb{R}^{n_{x_c}}$ is the state of the controller with some initial condition $x_c(0) = x_{c,0}$ and $\hat{x}_p(t)$ is the last plant state that has been received by the controller.

The sampled values of $\hat{u}$ and $\hat{x}_p$ are updated at discrete sampling instants $\svar_j, j \in \mathbb{N}_0$, generated by a self-triggered sampling mechanism that will be specified later. The updates are based on the current values of $u$ and $x$, i.e., $\hat{u}(\svar_j) = u(\svar_j)$ and $\hat{x}_p(\svar_j) = x_p(\svar_j)$. We will subsequently denote the value of the controller state at the last sampling instant as $\hat{x}_c$, i.e., $\hat{x}_c(\svar_j) = x_c(\svar_j)$. Between sampling instants, $\hat{x}_p, \hat{x}_c$ and $\hat{u}$ are thus constant, which corresponds to a zero-order-hold (ZOH) scenario.
We define the sampling-induced error as $e=\left[e_{x_p}^\top,e_{x_c}^\top\right]^\top \coloneqq\left[(\hat{x}_p-x_p)^\top,(\hat{x}_c-x_c)^\top\right]^\top$ and the combined state as $x \coloneqq \left[x_p^\top,x_c^\top\right]^\top$. Note that $e(t) \in \mathbb{R}^{n_e}$ and $x(t) \in \mathbb{R}^{n_x}$ for $n_x = n_e = n_{x_p} + n_{x_c}$.

The sampling instants are determined by a dynamic sampling mechanism that can be described as  $\svar_{j+1} \coloneqq \svar_j + \Gamma(x(\svar_j),\eta(\svar_j)),$
where $\Gamma:\mathbb{R}^{n_x}\times\mathbb{R}^{n_\eta} \rightarrow \left[\svar_{\min},\infty\right)$ for some $\svar_{\min} > 0$ to be specified later. Here $\eta \in \mathbb{R}^{n_\eta}$ is an internal state of the sampling mechanism. It is updated according to $	\eta(\svar_{j+1}) = S(\eta(\svar_j),x(\svar_j))$ for some $S:\mathbb{R}^{n_\eta}\times\mathbb{R}^{n_x} \rightarrow \mathbb{R}^{n_\eta}$,  which allows to take the past system behavior into account for determining transmission times.  

We define the timer variable $\tau$ with $\tau(t_j) = 0$ and $\dot{\tau}(t) = 1$ for $\svar_j \leq t < \svar_{j+1}$, which keeps track of the elapsed time since the last sampling instant and the auxiliary variable $\auxvar$ which encodes the next sampling interval. Using this, we can model the overall networked control system as a hybrid system $\mathcal{H}_{STC}$ according to \eqref{eq_sys_hyb}
with $\xi \coloneqq \left[x^\top,e^\top,\eta^\top,\tau,\auxvar\right]^\top,$
\begin{equation*}
F(\xi) \coloneqq \left(f(x,e,w)^\top,g(x,e,w)^\top,0,1,0\right)^\top,
\end{equation*}
with
\begin{align*}
f(x,e,w) = \begin{bmatrix}
f_p(x_p,g_c(x_c+e_{x_c},x_p+e_{x_p}),w)\\
f_c(x_c,x_p+e_{x_p}) 
\end{bmatrix}	
\end{align*}
and $g(x,e,w) = -f(x,e,w)$,
\begin{equation*}
G(\xi) \coloneqq \left(x^\top,0,S(\eta,x)^\top,0,\Gamma(x,\eta)\right)^\top,
\end{equation*}
and with 
\begin{equation}
\begin{split}
C := \left\lbrace \xi \in \mathbb{R}^{n_x+n_e+n_\eta+2} | \tau \leq \auxvar \right\rbrace\\
D := \left\lbrace \xi \in \mathbb{R}^{n_x+n_e+n_\eta+2} | \tau = \auxvar \right\rbrace.
\end{split}
\end{equation}

Note that jumps of the hybrid system~\eqref{eq_sys_hyb} correspond exactly to sampling instants of the self-triggered sampling mechanism and thus the sampling sequence $\svar_j, j \in \mathbb{N}_0$ corresponds exactly to the time indices when \eqref{eq_sys_hyb} jumps. Thus, we describe the hybrid time before the sampling at time $\svar_j$  by $\tvar_j \coloneqq  (\svar_j,j)$ and the hybrid time directly after the sampling at time $\svar_j$ by $\tvar_j^+ \coloneqq (\svar_j,j+1)$. Using this notation, the dynamic STC mechanism can be described by Algorithm~\ref{algo_stc}.
	\begin{algorithm}[tb]
		\caption[Caption for LOF]{Dynamic STC mechanism at sampling instant $t_j$}
		\label{algo_stc}
		\begin{algorithmic}[1]
			\STATE Measure $x_p(\tvar_j)$
			\STATE Set $\hat{x}_p(\tvar_j^+) = x_p(\tvar_j)$
			\STATE Set $\hat{u}(\tvar_j^+) = g_c(x_c(\tvar_j),x_p(\tvar_j))$
			\STATE Set $\auxvar(\tvar_j^+) = \Gamma(x(\tvar_j),\eta(\tvar_j))$ 
			\STATE Set $\eta(\tvar_j^+) = S(\eta(\tvar_j),x(\tvar_j))$
			\STATE Wait until $\svar_{j+1} = \svar_j + \auxvar(\tvar_j^+)$
		\end{algorithmic}
	\end{algorithm}

For simplicity, we assume that the self-triggered sampling mechanism is executed at the initial time $t_0 = 0$, i.e., there is a jump at $t = 0$ and we have that $\xi(0,1) = G(\xi(0,0))$. This corresponds to a restriction of the initial conditions of the hybrid system for $\tau(0,0)$ and $\auxvar(0,0)$ to $\auxvar(0,0) =\tau(0,0)$. Note that without this assumption, the first sampling instant might be not well-defined.

\subsection{Stability definitions and problem statement}
In this article, we will, depending on the disturbance signal $w(t)$, consider different stability notions. In the nominal case, i.e., if $w(t) = 0$ for all $t \geq 0$, we will be interested in guaranteeing uniform global asymptotic stability of the origin of $\mathcal{H}_{STC}$ according to the following definition, which is adapted to our setup from \cite{carnevale2007lyapunov}.
\begin{defi}
	\label{def_asym_stab}
	For the hybrid system~$\mathcal{H}_{STC}$ with $w(t) = 0$ for all $t \geq 0$, the set $\left\lbrace \left(x,e,\eta,\tau,\auxvar\right): x = 0, e= 0, \eta = 0 \right\rbrace$ is uniformly globally asymptotically stable (UGAS), if there exists $\beta \in \mathcal{K}\mathcal{L}\mathcal{L}$ such that, for each initial condition $x(0,0) \in \mathbb{R}^{n_x}$, $\eta(0,0) \in \mathbb{R}^{n_\eta}$, $e(0,0) \in \mathbb{R}^{n_e}$, $\tau(0,0) \in \mathbb{R}_{\geq 0}$ and $\auxvar(0,0) = \tau(0,0) $, and each corresponding solution
	\begin{equation}
	\label{eq_stab_bound}
	\abs{\begin{bmatrix}
		x(t,j)\\
		e(t,j)\\
		\eta(t,j)
		\end{bmatrix}} \leq \beta\left(\abs{\begin{bmatrix}
		x(0,0)\\
		e(0,0)\\
		\eta(0,0)
		\end{bmatrix}},t, j\right)
	\end{equation}
	for all $(t,j)$ in the solutions domain. 
\end{defi} 

Since the disturbance signal will in practice be nonzero, stability notions that take into account the effects of the disturbance need to be considered.
An important stability notion for perturbed systems is input-to-state stability, which we define for the hybrid system $\mathcal{H}_{STC}$ as follows.
\begin{defi}
	\label{def_iss}
	The hybrid system $\mathcal{H}_{STC}$ is input-to-state stable (ISS), if there exist $\beta \in \mathcal{K}\mathcal{L}\mathcal{L}$ and $\psi \in \mathcal{K}_\infty$, such that, for each initial condition $x(0,0) \in \mathbb{R}^{n_x}$, $\eta(0,0) \in \mathbb{R}^{n_\eta}$, $e(0,0) \in \mathbb{R}^{n_e}$, $\tau(0,0) \in \mathbb{R}_{\geq 0}$ and $\auxvar(0,0) = \tau(0,0) $, and each corresponding solution 
	 	\begin{equation}
	 	\abs{\begin{bmatrix}
	 		x(t,j)\\
	 		e(t,j)\\
	 		\eta(t,j)
	 		\end{bmatrix}} \leq \beta\left(\abs{\begin{bmatrix}
	 		x(0,0)\\
	 		e(0,0)\\
	 		\eta(0,0)
	 		\end{bmatrix}},t,j\right) + \psi(\norm{w}_\infty)
	 	\end{equation}
	 	for all $(t,j)$ in the solutions domain. 	
\end{defi}
Note that if $\mathcal{H}_{STC}$ is ISS, this implies UGAS of the set $\left\lbrace \left(x,e,\eta,\tau,\auxvar\right): x = 0, e= 0, \eta = 0 \right\rbrace$ if $w(t) = 0$ for all $t \geq 0$.
If a bound on the disturbance is known, i.e., if $\abs{w(t)} \in \mathcal{W}\forall t \geq 0$ for a compact set $\mathcal{W}$, then we can instead consider stability of a sublevel set of a radially unbounded positive definite function $V:\mathbb{R}^{n_x}\rightarrow \mathbb{R}_{\geq 0}$ according to the following definition.
\begin{defi}
	\label{def_set_stab}
	Consider $c_{\max}\geq c_w>0$. For the hybrid system~$\mathcal{H}_{STC}$ with $\abs{w(t)}\in \mathcal{W}$ for all $t \geq 0$, the set $\mathcal{R}\coloneqq\left\lbrace \left(x,e,\eta,\tau,\auxvar\right): V(x) \leq c_w \right\rbrace$ is robustly asymptotically stable (RAS) with region of attraction (ROA) $\mathcal{X}_{c_{\max}} \coloneqq \left\lbrace x|V(x) \leq c_{\max} \right\rbrace \supseteq \mathcal{R}$, if there exists $\beta \in \mathcal{K}\mathcal{L}\mathcal{L}$ such that, for each initial condition $x(0,0) \in \mathcal{X}_{c_{\max}}$, $\eta(0,0) \in \mathbb{R}^{n_\eta}$, $e(0,0) \in \mathbb{R}^{n_e}$, $\tau(0,0) \in \mathbb{R}_{\geq 0}$ and $\auxvar(0,0) = \tau(0,0) $, and each corresponding solution
	\begin{equation}
		\label{eq_stab_bound_set}
		\begin{split}	
			&V(x(t,j)) \\
			\leq& c_w+\beta\left(\max\left\lbrace V(x(0,0)) - c_w, 0 \right\rbrace + \abs{\begin{bmatrix}
					e(0,0)\\
					\eta(0,0)
			\end{bmatrix}},t, j\right)\end{split}
	\end{equation}
	for all $(t,j)$ in the solutions domain. 
\end{defi} 

Hereby, considering the compact region of attraction $\mathcal{R}$ allows us to consider local results, which may be beneficial for nonlinear systems, where global stability results can not always be obtained. 

In this article, our goal is to design functions $\Gamma$ and $S$ for $\mathcal{H}_{STC}$, such that for the resulting dynamic STC mechanisms, the time spans between sampling instants are maximized whilst at the same time the stability properties mentioned above are ensured. To this end,
 the dynamic variable $\eta$ shall be exploited.
\section{Hybrid Lyapunov functions for dynamic STC}
\label{sec_concept}
In this section, we present a general framework for dynamic STC based on hybrid Lyapunov functions. This framework includes different particular dynamic STC mechanisms. We assume in this section and in Section~\ref{sec_spec} that neither the disturbance signal nor a bound on it is known. We present in the first subsection some preliminaries on hybrid Lyapunov functions and then discuss in the second subsection how they can be used to determine sampling instants for dynamic STC in a general framework.
 The framework allows different particular choices for $\Gamma$ and $S$ which lead to different dynamic STC mechanisms with different properties of the closed-loop system. Specific choices for $\Gamma$ and $S$ with guarantees for ISS will be discussed in Section~\ref{sec_spec}.

\subsection{Hybrid Lyapunov functions}
We make the following assumption based on \cite[Condition IV.6]{heemels2010networked}, that can be used to derive a hybrid Lyapunov function.
\begin{asum}
	\label{asum_hybrid_lyap}
	 There exist a locally Lipschitz function $W:\mathbb{R}^{n_e} \rightarrow \mathbb{R}_{\geq0}$, a locally Lipschitz function $V:\mathbb{R}^{n_x} \rightarrow \mathbb{R}_{\geq0}$, a continuous function $H:\mathbb{R}^{n_x}\times\mathbb{R}^{n_e}\times\mathbb{R}^{n_w} \rightarrow \mathbb{R}_{\geq0}$, constants $L, \gamma\in \mathbb{R}_{>0}$, $\epsilon\in \mathbb{R}$, and  $\underline{\alpha}_W$, $\overline{\alpha}_W, \underline{\alpha}_V, \overline{\alpha}_V,\alpha_w \in \mathcal{K}_\infty$  such that for all $e\in\mathbb{R}^{n_e}$,
	\begin{equation}
		\label{eq_w_bound}
		\underline{\alpha}_W(\abs{e}) \leq W(e) \leq \overline{\alpha}_W(\abs{e}),
	\end{equation}
	for all $x \in\mathbb{R}^{n_x}$,
	\begin{equation}
		\label{eq_V_bound_K}
		\underline{\alpha}_V(\abs{x}) \leq V(x) \leq \overline{\alpha}_V(\abs{x}),
	\end{equation}
	and for all  $x \in \mathbb{R}^{n_x}, w \in \mathbb{R}^{n_w} $ and almost all $e\in\mathbb{R}^{n_e},$ 
	\begin{equation}
		\left\langle \frac{\partial W(e)}{\partial e},g(x,e,w)\right\rangle \leq L W(e) + H(x,e,w).  \label{eq_w_est}
	\end{equation}
	Moreover, for all $e \in \mathbb{R}^{n_e}, w \in \mathbb{R}^{n_w} $ and almost all $x \in \mathbb{R}^{n_x}$,  
	\begin{equation}
		\begin{split}
			&\left\langle \nabla V(x),f(x,e,w) \right\rangle\\
			\leq &- \epsilon V(x) -H^2(x,e,w)
			+ \gamma^2 W^2(e) + \alpha_w\left(\abs{w}\right).
		\end{split}
		\label{eq_v_desc_hybrid}
	\end{equation}
\end{asum}
Note that there are some differences between Assumption~\ref{asum_hybrid_lyap} and \cite[Condition IV.6]{heemels2010networked}. In \eqref{eq_v_desc_hybrid}, we add the additional term $- \epsilon V(x)$. For $\epsilon > 0$, this term requires that the system with continuous feedback is exponentially stable. For $\epsilon < 0$, it allows potentially to chose smaller values for $\gamma$, which we will exploit subsequently to maximize the time between triggering instants. Furthermore, we use here a $\mathcal{K}_\infty$ function $\alpha_w\left(\abs{w}\right)$ instead of $\theta^p \abs{w}^p$ for some $\theta > 0$. This relaxation is possible since we investigate input-to-state stability instead of particular $\mathcal{L}_p$-gains, as it was done in \cite{heemels2010networked}. Approaches for verifying Assumption~\ref{asum_hybrid_lyap} can, e.g., be adapted from \cite{hertneck20simple}.

Note also that Assumption~\ref{asum_hybrid_lyap} can hold simultaneously for different choices of $\epsilon, \gamma$ and $L$. If we can find one parameter set for which the assumption holds, then we will typically also be able to find many different parameter sets.

To determine sampling intervals, the proposed STC framework will employ a bound on the evolution of $V(x)$, which is adapted from \cite{hertneck20stability}. This bound is based on the function
\begin{equation}
	T_{\max}(\gamma,\lvar ) \coloneqq \begin{cases}\vspace{1mm}
	\frac{1}{\lvar r} \mathrm{arctan}(r) & \gamma > \lvar \\ \vspace{1mm}
	\frac{1}{\lvar } & \gamma = \lvar \\
	\frac{1}{\lvar r} \mathrm{arctanh}(r) &\gamma < \lvar 
	\end{cases}
\end{equation}
where
\begin{equation}
	\label{eq_def_r}
	r\coloneqq\sqrt{\abs{
	\left(\frac{\gamma}{\lvar }\right)^2-1}},
\end{equation} 
that was originally used in \cite{nesic2009explicit} to determine the maximum allowable transmission interval. The choice of the parameter $\Lambda$ will be discussed subsequently. A visualization of $T_{\max}(\gamma,\Lambda)$ is given in Figure~\ref{fig_surf}. 

\begin{figure}
	\resizebox{.9\columnwidth}{!}{\includegraphics{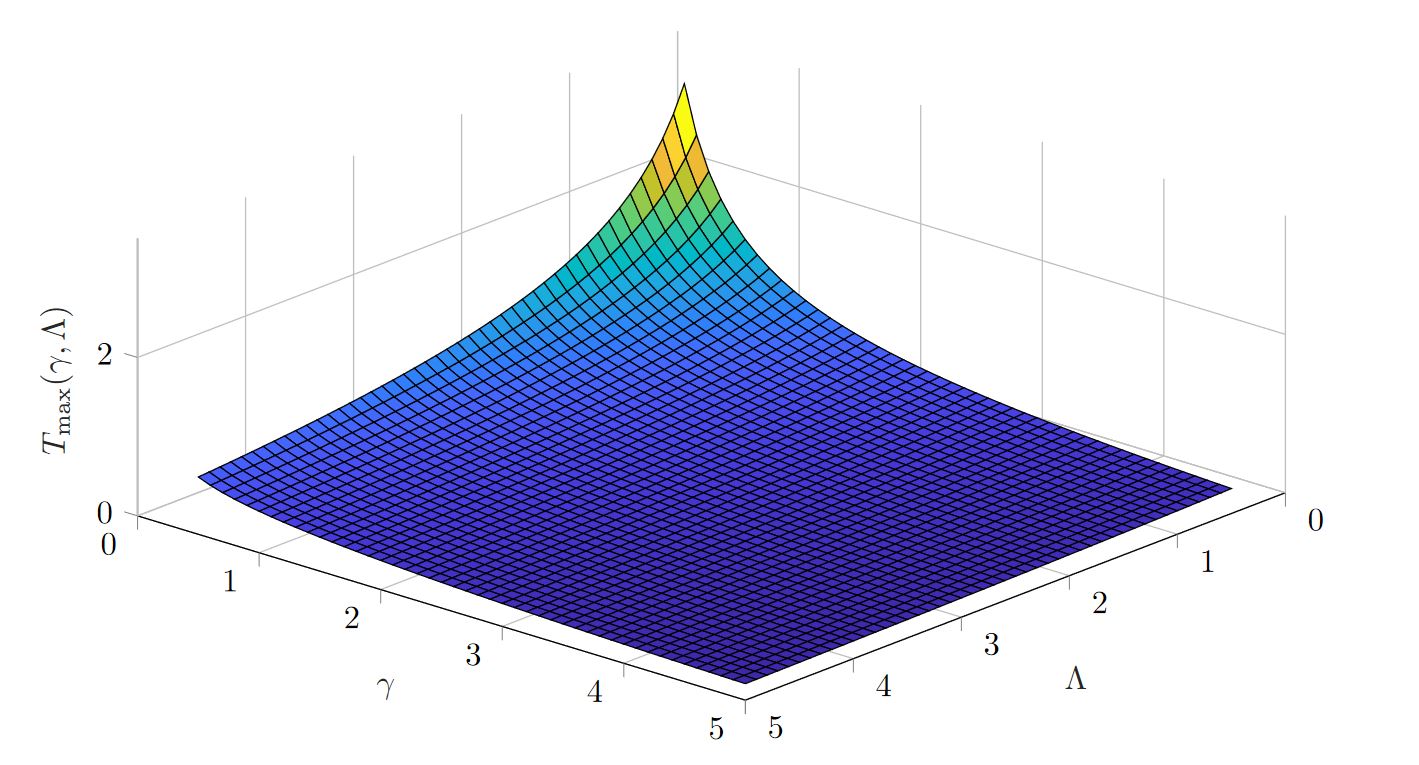}}
	\caption{Surface plot of $T_{\max}(\gamma,\Lambda)$.}
	\label{fig_surf}
\end{figure}

 We adapt from \cite[Proposition 12]{hertneck20stability} the following result.
\begin{prop}
	\label{prop_hybrid}
	Consider the hybrid system $\mathcal{H}_{STC}$ at sampling instant $\tvar_j^+$ for $j \in \mathbb{N}_0$.  Let Assumption~\ref{asum_hybrid_lyap} hold 
	for $\gamma, \epsilon$ and $L$.
	Moreover, let $0 < \auxvar(\tvar_j^+) < T_{\max} (\gamma,\lvar )$  for $\lvar  >0$.	
	Consider 
	\begin{equation}
	\label{eq_def_u}
	U(\xi) \coloneqq V(x)+\gamma \phi(\tau) W^2(e),
	\end{equation}
	where 	 $\phi : [0,\auxvar(\tvar_j^+)] \rightarrow \mathbb{R}$ is the solution to
	\begin{equation}
	\label{eq_def_phi}
	\dot{\phi} = -2\lvar \phi-\gamma(\phi^2+1),~ \phi(0) = \lambda^{-1}
	\end{equation}
	for some sufficiently small $\lambda \in \left(0,1\right)$.
	Then, for all $\svar_j \leq t \leq \svar_j+\auxvar(\tvar_j^+)$, it holds that	
	\begin{equation}
	\begin{split}
	&V(x(\svar,j+1))\\
	 \leq& U(\xi(\svar,j+1))
	\leq	e^{\max \left\lbrace -\epsilon, 2(L - \lvar ) \right\rbrace (t-\svar_j)}V(x(\tvar_j))\\
	&+ \int_{t_j}^{t} e^{\max \left\lbrace -\epsilon, 2(L - \lvar ) \right\rbrace (t-\tau)}
	 \alpha_w\left(\abs{w}\right)  d\tau .
	\end{split}
	\label{eq_prop_hybrid1}
	\end{equation}
\end{prop}
\begin{proof}
	The proof is given in Appendix~\ref{proof_prop_hyb}.
\end{proof}				
Proposition~\ref{prop_hybrid} delivers for parameters $\epsilon,\gamma, L$  and a $\mathcal{K}_\infty$ function $\alpha_w$ that satisfy Assumption~\ref{asum_hybrid_lyap} an upper bound on the evolution of $U(\xi)$ and thus since $U(\xi) \geq V(x)$ also an upper bound on the evolution of $V(x)$. This bound is valid, if the time between two sampling instants is bounded by $T_{\max}(\gamma,\Lambda)$ for some $\Lambda > 0$. The actual bound depends on the parameters from Assumption~\ref{asum_hybrid_lyap} and on $\Lambda$. Particularly, if $\epsilon > 0$ and $\Lambda > L$, then the bound is exponentially decreasing in the nominal case (i.e. for $w(t) = 0$ for $\svar_j \leq t \leq \svar_{j} + \auxvar\left(\tvar_j^+\right)$). In contrast, if $\epsilon < 0$ and $\Lambda < L$, then the bound is increasing. However, the admissible time between sampling instants $T_{\max}(\gamma,\Lambda)$ decreases when $\gamma$ and $\Lambda$ are increased.  Note that smaller values for $\epsilon$ require larger values for $\gamma$ for \eqref{eq_v_desc_hybrid} in Assumption~\ref{asum_hybrid_lyap} to hold. We thus observe in Proposition~\ref{prop_hybrid} a trade-off between the admissible time between sampling instants and the growth of the bound on $V(x)$. Particularly, if the time between sampling instants is small, then we will be able to chose $\epsilon$ and $\Lambda$ large and thus obtain an exponentially decreasing bound on $V(x)$ for the nominal case. In contrast, if the time between two sampling instants is large, then we need to chose $\Lambda$ and $\epsilon$ small to be able to derive a bound on $V(x)$, which has the effect that this bound may be increasing. Next, we discuss how the bound
 on $V(x)$ 
 can be used to determine sampling instants.
\subsection{Using hybrid Lyapunov functions to determine sampling instants}
\label{sec_def_algo}
For periodic time-triggered sampling, to determine sampling instants based on Proposition~\ref{prop_hybrid}, we would have to choose $\epsilon$ and $\Lambda$ such that the bound on $V(x)$ is decreasing in order to obtain stability guarantees. The value of $\epsilon$ would determine the smallest possible values of $\gamma$ and $L$ and as a consequence also the maximum allowable sampling interval, for which we would be able to guarantee stability. 

However the upper bound on $V(x)$ from Proposition~\ref{prop_hybrid} is typically conservative, i.e., $V(x)$ decreases in the nominal case often more strongly than guaranteed, and the effect of the disturbance is often not the worst-case effect, as is considered in Proposition~\ref{prop_hybrid}. It may therefore be reasonable to allow sometimes a certain increase of $V(x)$ as long as an average decrease of $V(x)$ can be guaranteed or a comparable condition, that ensures desired stability properties, is still satisfied. We will now describe how this can be exploited by a dynamic STC mechanism to maximize sampling intervals.

Suppose the condition, that the dynamic STC mechanism shall ensure when choosing at time $\tvar_j$ the next sampling instant $\svar_{j+1}$, has the form that in the nominal case
\begin{equation}
	\label{eq_cond_dec}	V(x(t,j+1)) \leq e^{-\epsilon_\refer(t - \svar_j)} C(x(\tvar_j),\eta(\tvar_j))
\end{equation}
has to hold  for $\svar_j \leq t \leq \svar_{j+1}$, a tunable constant $\epsilon_\refer > 0$ and a function $C:\mathbb{R}^{n_x}\times\mathbb{R}^{n_\eta}\rightarrow\mathbb{R}_{\geq 0}$ that will be specified later. The condition is stated for the nominal case since we assume in this section that the disturbance signal $w(t)$ may be arbitrary and unknown, and thus, the dynamic STC mechanism cannot take it explicitly into account for determining sampling instants. Nevertheless, we will later derive ISS guarantees for the STC mechanism despite the disturbance for different choices of $C(x,\eta)$. 

 The dynamic STC mechanism can now use the bound on $V(x)$ from Proposition~\ref{prop_hybrid} to choose the next sampling instant, i.e., to determine at time $\tvar_j$ a preferably large value for $\auxvar(\tvar_j^+) = \svar_{j+1} - \svar_j$ such that \eqref{eq_cond_dec} holds in the nominal case. For that, suppose the dynamic STC mechanism has access to $n_\tpar$ different parameter sets $\epsilon_i,\gamma_i,L_i$, $i\in\left\lbrace1,\dots,n_\tpar\right\rbrace$, for which Assumption~\ref{asum_hybrid_lyap} holds for the same functions $\alpha_w$ and $V(x)$. We will later require that at least for one of the parameter sets, to which we assign the index 1, $\epsilon_1 > 0$ holds. For all other parameter sets, $\epsilon_i$ may be negative. In the nominal case, \eqref{eq_cond_dec} holds due to Proposition~\ref{prop_hybrid}, if for one of the parameter sets
\begin{equation}
\label{eq_dec_gen}
	\begin{split}
		&e^{\max \left\lbrace -\epsilon_i, 2(L_i - \lvar_i ) \right\rbrace (t-\svar_j)}V(x(\tvar_j))\\
		 \leq& e^{-\epsilon_\refer(t - \svar_j)} C(x(\tvar_j),\eta(\tvar_j))
	\end{split}
\end{equation} 
holds for some $\lvar_i > 0$ and 
\begin{equation}
\label{eq_bound_aux}
	t-\svar_j \leq \svar_{j+1}-\svar_j = \auxvar(s_j^+)< T_{\max}(\gamma_i,\lvar_i).
\end{equation} 
Thus, the dynamic STC mechanism needs to search at hybrid time $s_j$ for a preferably large value for $\auxvar(s_j^+)$ such that \eqref{eq_dec_gen} and \eqref{eq_bound_aux} hold for some $i \in \left\lbrace1,\dots,n_\tpar\right\rbrace$ and $t_j \leq t \leq t_j + \auxvar(s_j^+)$. For an efficient search, we make two simplifications. First, we replace condition \eqref{eq_bound_aux} by 
\begin{equation}
	\label{eq_dec_gen_mod}
	\auxvar(s_j^+)\leq \delta T_{\max}(\gamma_i,\lvar_i)
\end{equation} 
for some\footnote{Typically, $\delta$ will be chosen close to 1 (e.g. 0.999) in order to obtain a preferably large sampling interval.} $\delta \in \left(0,1\right)$. Second, we fix\footnote{Note that choosing $\Lambda_i$ larger would not provide any advantage, whilst smaller choices could in some situations be advantageous. Since the expected advantage is typically minor, we omit it in the STC mechanism to reduce the computational complexity. A line search could be included to exploit different values for $\Lambda_i$.}  $\Lambda_i = \max \left\lbrace L_i+\frac{\epsilon_i}{2}, 1-\delta \right\rbrace$.  Here the (typically small) positive value of $1-\delta$ avoids that $\Lambda_i \leq 0$. Note that Proposition~\ref{prop_hybrid} still applies despite the simplifications. 

The second simplification allows us to rewrite \eqref{eq_dec_gen} as
\begin{equation}
\left(-\epsilon_i +\epsilon_\refer\right)(t-\svar_j)
\leq \log\left(\frac{C(x(\tvar_j),\eta(\tvar_j))}{V(x(\tvar_j))}\right).
\label{eq_ln}
\end{equation}

Now suppose that $C(x(\tvar_j),\eta(\tvar_j)) \geq  V(x(\tvar_j))$. In this case, maximizing  $\auxvar(\tvar_j^+)$ for a given $i \in \left\lbrace1,\dots, n_\tpar\right\rbrace$ such that \eqref{eq_dec_gen_mod} and \eqref{eq_ln} hold is straightforward. In particular, if $-\epsilon_i + \epsilon_\refer > 0$, then we obtain
\begin{align*}
\nonumber
\auxvar(\tvar_j^+)  =& \min\left\lbrace \delta T_{\max}(\gamma_i,\Lambda_i)\vphantom{\frac{\log(\windowC)-\log(V(x(\tvar_j)))}{\max\left\lbrace -\epsilon_i,2(L_i-\lvar_i ) \right\rbrace + \epsilon_\refer}}\right.,\\
&\left.\frac{\log(C(x(\tvar_j),\eta(\tvar_j)))-\log(V(x(\tvar_j)))}{ -\epsilon_i + \epsilon_\refer} \right\rbrace
\end{align*} as the maximum value. Otherwise, i.e., if $ -\epsilon_i + \epsilon_\refer \leq 0$, we obtain the maximum value $\auxvar(\tvar_j^+) = \delta T_{\max}(\gamma_i,\Lambda_i)$. 

The case that  $C(x(\tvar_j),\eta(\tvar_j)) <  V(x(\tvar_j))$ is typically not relevant and therefore omitted\footnote{Note that it is addressed in the preliminary study \cite{hertneck21dynamic}.} by the dynamic STC mechanism for simplicity. In this case, it will use a fall-back strategy detailed subsequently.

An efficient search for a preferably large value of $\auxvar(s_j^+)$ is thus possible for any given parameter set that satisfies Assumption~\ref{asum_hybrid_lyap}. The STC mechanism can now simply iterate over all parameter sets for $i \in \left\lbrace2,\dots,n_\tpar\right\rbrace$ and use the largest value for $\auxvar(\tvar_j^+)$ for which a guarantee can be obtained that \eqref{eq_cond_dec} holds if $\svar_{j+1} - \svar_j = \auxvar(\tvar_j^+)$. It may however happen that a guarantee that \eqref{eq_cond_dec} holds cannot be obtained for any $i \in \left\lbrace2,\dots,n_\tpar\right\rbrace$. In this case, the STC mechanism uses a fall-back strategy that exploits that $\epsilon_1 > 0$. In particular, it chooses then $\auxvar(\tvar_j^+) = \delta T_{\max}\left(\epsilon_1,L_1+\frac{\gamma_1}{2}\right)$, for which it follows in the nominal case from Proposition~\ref{prop_hybrid} for $\svar_{j} \leq t \leq \svar_{j+1}$ that
\begin{equation*}
	V(x(t,j+1)) \leq e^{-\epsilon_1\left(t - \svar_j\right)} V(x(\svar_j)),
\end{equation*}   
which can be employed to obtain stability guarantees.

	\begin{algorithm}[tb]
		\caption{Computation of $\Gamma(x,\eta)$ for the dynamic STC framework for some $\delta \in \left(0,1\right)$ and given $C(x,\eta)$. }
		\label{algo_trig_window}
		\begin{algorithmic}[1]
			\STATE $V \leftarrow V(x)$, $C \leftarrow C(x,\eta)$ %
			\STATE $\bar h \leftarrow \delta T_{\max}\left(\gamma_\itilde,L_\itilde+\frac{\epsilon_\itilde}{2}\right)$ \label{line_fallback}
			\FOR{\text{\bf each} $i \in \left\lbrace2,\dots,n_\tpar\right\rbrace$ } 
			\STATE $\Lambda_i \leftarrow \max \left\lbrace L_i + \frac{\epsilon_i}{2},(1-\delta) \right\rbrace$
			\IF{$C \geq V$} \label{line_for_start} %
			\IF{$-\epsilon_i+\epsilon_\refer > 0$}
			\STATE $\bar h_i \leftarrow \min\left\lbrace \delta T_{\max}(\gamma_i,\Lambda_i),				\frac{\log(C)-\log(V)}{ -\epsilon_i+ \epsilon_\refer} \right\rbrace$ \label{line_hi}
			\ELSE
			\STATE $\bar h_i \leftarrow \delta T_{\max}(\gamma_i,\Lambda_i)$
			\ENDIF
			\ELSE
			\STATE $\bar h_i \leftarrow 0$
			\ENDIF\label{line_for_end}
			\IF{$\bar h_i > \bar h$}
			\STATE $\bar h \leftarrow \bar h_i$\label{line_h_update}
			\ENDIF
			\ENDFOR 			
			\STATE $\Gamma(x,\eta) \leftarrow \bar{h}$
		\end{algorithmic}
	\end{algorithm}
In Algorithm~\ref{algo_trig_window}, the procedure to determine a preferably large value for $\auxvar(s_j^+)$ 
such that \eqref{eq_cond_dec} holds in the nominal case is summarized. It will later serve for the dynamic STC mechanism as an implicit definition of the function $\Gamma(x,\eta)$ for given $C(x,\eta)$. The algorithm first computes the values of $V(x)$ and $C(x,\eta)$. Then it sets $\bar h$ to the minimum value for the fall-back strategy, that is given by $\bar h=T_{\max}\left(\gamma_1,L_1+\frac{\epsilon_1}{2}\right)$. After that, an iteration over all other parameter sets is started. For each parameter set, the algorithm determines a preferably large value for which \eqref{eq_cond_dec} holds if $\auxvar(\tvar_j^+)$ is chosen accordingly. If $\bar h$ is smaller than that value, then $\bar h$ is updated accordingly. Thus, after the iteration, the variable $\bar h$ contains a preferably large value for $\auxvar(\tvar_j^+)$, for which it is guaranteed that $\eqref{eq_cond_dec} $ holds in the nominal case or it is set according to the fall-back strategy.

In the next section, we will present particular choices for $C(x,\eta)$ and for $S(\eta,x)$ that lead, together with the implicit definition of $\Gamma(x,\eta)$ in Algorithm~\ref{algo_trig_window}, to particular dynamic STC mechanisms with different closed-loop properties. 
\section{Specific dynamic STC mechanisms with ISS guarantees}
\label{sec_spec}
In this section, we present 
three different dynamic STC mechanisms that are based on the implicit definition of $\Gamma(x,\eta)$ in Algorithm~\ref{algo_trig_window}. The first two mechanisms will use different (time-varying) linear filters for the values of the Lyapunov function $V(x)$ at past sampling instants to determine the next sampling instant. The third one uses instead a time-dependent but state-independent reference function.
Based on Assumption~\ref{asum_hybrid_lyap}, we derive for all three mechanisms guarantees for ISS, which also implies UGAS of the origin in the nominal case. 
\subsection{Dynamic STC based on an FIR filter  for the Lyapunov function}
\label{subsec_fir}
In this subsection, we present a dynamic STC mechanism that uses a time-varying FIR filter for the values of the Lyapunov function $V(x)$ at past sampling instants to determine the next sampling instant. Note that a preliminary version of this mechanism has been presented in \cite{hertneck21dynamic}. 
The mechanism uses Algorithm~\ref{algo_trig_window} to choose the next sampling instant $\svar_{j+1}$ at time $\tvar_j$ such that for some $m>1$,
\begin{equation}
	V(x(\tvar_{j+1})) \leq e^{-\epsilon_\refer(\svar_{j+1} - \svar_j)} \frac{1}{m} \sum_{k = j-m+1}^{j} e^{\epsilon_\refer \left(t_j - t_k\right)} V(x(s_k))
	\label{eq_fir_goal}
\end{equation}
would hold in the nominal case, i.e., such that the Lyapunov function at the next sampling instant would be bounded by a discounted average of the values of the Lyapunov function at the past $m$ sampling instants. To implement this choice in the setup from Section~\ref{sec_setup} and with Algorithm~\ref{algo_trig_window}, we set $n_\eta = m-1$ as the dimension of the dynamic variable and define the update rule of the dynamic variable as 
\begin{equation}
\label{eq_S_window}
S(\eta,x) = \begin{pmatrix}
e^{-\epsilon_\refer\Gamma(x,\eta)}\eta_2\\
\vdots\\
e^{-\epsilon_\refer\Gamma(x,\eta)}\eta_{m-1}\\
e^{-\epsilon_\refer\Gamma(x,\eta)}V(x)
\end{pmatrix},
\end{equation}
where $\Gamma(x,\eta)$ is defined by Algorithm~\ref{algo_trig_window} for a function $C(x,\eta)$ that is still to be determined. Note that  $t_{j+1}-t_{j} = \auxvar(\tvar_j^+) = \Gamma(x(\tvar_j),\eta(\tvar_j))$. Hence, for this choice of $S(\eta,x)$,
\begin{equation}
	\label{eq_fir_eta}
	\eta_k(s_j) = e^{-\epsilon_\refer (t_j-t_{j-m+k})} V(x(s_{j-m+k}))	
\end{equation}
holds if $k > m-j$, which implies
\begin{equation*}
V(x(\tvar_j))+\sum_{k=1}^{m-1} \eta_k(\tvar_j) = \sum_{k = j-m+1}^{j} e^{\epsilon_\refer \left(t_j - t_k\right)} V(x(s_k))
\end{equation*}
 for $j > m$. 
 Then, \eqref{eq_cond_dec} is equal to \eqref{eq_fir_goal} if we choose
\begin{equation}
	\label{eq_fir_C}
	C(x,\eta) = \frac{1}{m} \left(V(x)+ \sum_{k=1}^{m-1} \eta_k\right).
\end{equation}
 For $k \leq m-j$, the value of $\eta_k$ is determined by the initial condition $\eta(0,0)$ and does not influence the stability properties.
Therefore, we define the function $\Gamma(x,\eta)$ for the dynamic STC mechanism, that is based on a discounted average of the Lyapunov function, implicitly by Algorithm~\ref{algo_trig_window} with $C(x,\eta)$ according to \eqref{eq_fir_C}.
This leads us to the following result.
\begin{theo}
	\label{prop_fir}
	Assume there are $n_\tpar$ different parameter sets $\epsilon_i, \gamma_i, L_i$, $i \in \left\lbrace 1,\dots,n_\tpar\right\rbrace$, for which Assumption~1 holds with the same functions $V$ and $\alpha_w$. Let $\epsilon_1>0$. Further, consider $\mathcal{H}_{STC}$ with $S(\eta,x)$ and $\Gamma(x,\eta)$ defined according to \eqref{eq_S_window} and by Algorithm~\ref{algo_trig_window} with $C(x,\eta)$ according to \eqref{eq_fir_C} and some $\delta \in 
	\left(0,1\right)$. Then $\svar_{j+1}- \svar_j \geq t_{\min} \coloneqq \delta T_{\max}\left(\gamma_1,L_1+\frac{\epsilon_1}{2}\right)\forall j\in\mathbb{N}_0$ and $\mathcal{H}_{STC}$ is ISS.
\end{theo}
\begin{proof}
	The proof is given in Appendix~\ref{proof_prop_fir}.
\end{proof}
\begin{rema}
	Theorem~\ref{prop_fir} directly implies UGAS of the set $\left\lbrace \left(x,e,\eta,\tau,\auxvar\right): x = 0, e= 0, \eta = 0 \right\rbrace$ in the nominal case, as this is a direct consequence of the definition of ISS.	
\end{rema}
\begin{rema}
	The combination of \eqref{eq_S_window} and \eqref{eq_fir_C} corresponds to a (time-varying) FIR filter for the Lyapunov function $V(x)$ evaluated at sampling instants. The $e^{-\epsilon_\refer \Gamma(x,\eta)}$ terms are included here to determine the convergence speed of the dynamic STC mechanism in the nominal case. Instead, a constant factor could be used as it was done in the preliminary study \cite{hertneck21dynamic}. However, then the convergence behavior of the filter would be influenced by the time between sampling instants.
\end{rema}

\subsection{Dynamic STC based on an IIR filter for the Lyapunov function}
\label{subsec_iir}
In this subsection, we present a second approach to choose the dynamics of the dynamic variable. While the approach from the previous subsection was based on an FIR filter, we consider in this subsection an approach that is based on a (time-varying) IIR filter. 
To implement the IIR filter, we set $n_\eta = 1$ and choose 
\begin{equation}
\label{eq_S_iir}
	S(\eta,x) = e^{-\epsilon_\refer \Gamma(x,\eta)}\left(r_1 \eta + r_2 V(x)\right)
\end{equation}
where  $r_1 \in \mathbb{R}_{>0}$ and $r_2 \in \mathbb{R}_{>0}$ are some constants satisfying $r_1+r_2 \leq 1$ and $\Gamma(x,\eta)$ is again defined by Algorithm~\ref{algo_trig_window} for a function $C(x,\eta)$ that will be determined next. The trigger decision is made for this STC mechanism such that 
\begin{equation}
	V(x(\tvar_{j+1})) \leq e^{-\epsilon_\refer\left(\svar_{j+1}-\svar_j\right)} \eta(\tvar_j)
\end{equation}
holds in the nominal case, i.e., such that the value of the Lyapunov function at the next sampling instant is bounded by the state of the filter. This can be achieved by using Algorithm~\ref{algo_trig_window} with 
\begin{equation}
	\label{eq_iir_C}
	C(x,\eta) = \eta.
\end{equation}
Thus, we define the function $\Gamma(x,\eta)$ for the dynamic STC mechanism, that is based on a  time varying IIR filter, implicitly by Algorithm~\ref{algo_trig_window} with $C(x,\eta)$ according to \eqref{eq_iir_C}. We obtain the following result. 
\begin{theo}
	\label{theo_iir}
	Assume there are $n_\tpar$ different parameter sets $\epsilon_i, \gamma_i, L_i$, $i \in \left\lbrace 1,\dots,n_\tpar\right\rbrace$, for which Assumption~1 holds with the same functions $V$ and $\alpha_w$. Let $\epsilon_1,r_1,r_2>0$ and $r_1+r_2 \leq 1$. Further, consider $\mathcal{H}_{STC}$ with $S(\eta,x)$ and $\Gamma(x,\eta)$ defined according to \eqref{eq_S_iir} and by Algorithm~\ref{algo_trig_window} with $C(x,\eta)$ according to \eqref{eq_iir_C} and some $\delta \in 
	\left(0,1\right)$. Then $\svar_{j+1}- \svar_j \geq t_{\min} \coloneqq \delta T_{\max}\left(\gamma_1,L_1+\frac{\epsilon_1}{2}\right)$ and $\mathcal{H}_{STC}$ is ISS.
\end{theo}
\begin{proof}
	The proof is given in Appendix~\ref{proof_theo_iir}
\end{proof}
\begin{rema}
	The update of $\eta$ according to \eqref{eq_S_iir} can be interpreted as a time-varying IIR filter for the values of the Lyapunov function at past sampling instants. Similar as in the previous subsection, the additional term $e^{-\epsilon_\refer \Gamma(x,\eta)}$ is  included to determine the convergence speed of the closed-loop system with the dynamic STC mechanism in the nominal case. This term could also be replaced by a constant term. However, then the convergence behavior of the filter would depend on the time between sampling instants, which may be undesired. 
\end{rema}
\begin{rema}
	The parameters $r_1$ and $r_2$ can be used to tune the behavior of the IIR filter. The condition $r_1+r_2 \leq 1$ ensures that the interconnection of system and filter is stable.
\end{rema}
\subsection{Dynamic STC based on a time-dependent reference function}
\label{subsec_ref}
In the two previous subsections, we have presented two dynamic STC mechanisms that are based on a linear filter for the Lyapunov function $V(x)$ at past sampling instants and that are thus based on past system states. In this subsection, we will present a different approach, that instead bounds the Lyapunov function $V(x)$ at the next sampling instant by a reference function that depends only on time and on the initial state of the NCS. It is thus independent of the actual state evolution of the system. Such a reference function-based approach may, e.g., be advantageous for setpoint changes. In particular, the goal of the dynamic STC mechanism that we present in this subsection is to ensure  for a function $ V_\refer:\mathbb{R}_{\geq 0}\times\mathbb{R}^{n_x}\rightarrow\mathbb{R}_{\geq 0}$, that
\begin{equation}
	\label{eq_ref_goal}
	V(x(t,j)) \leq V_\refer(t,x(0,0))
\end{equation}
holds in the nominal case for all $(t,j) \in \text{dom}~\xi$. We assume for simplicity
that $\eta(0,0) = V(x(0,0))$, which is not restrictive if the initial value of the dynamic variable can be set by the user, and focus on the specific reference function choice $V_\refer(t,x(0,0))\coloneqq e^{-\epsilon_\refer t} V(x(0,0))$ for $\epsilon_\refer > 0$. This choice can be implemented by using the dynamic variable with $n_\eta=1$ and
\begin{equation}
	\label{eq_S_ref}
	S(\eta,x) = e^{-\epsilon_\refer\Gamma(x,\eta)} \eta,
\end{equation}
where $\Gamma(x,\eta)$ is again defined by Algorithm~\ref{algo_trig_window} for a function $C(x,\eta)$ that will be determined next. Recall that we want to choose sampling instants such that \eqref{eq_ref_goal} holds.   
This can be achieved by using Algorithm~\ref{algo_trig_window} with
\begin{equation}
	\label{eq_ref_C}
	C(x,\eta) = \eta.
\end{equation}
Hence, the function $\Gamma(x,\eta)$ is defined for the dynamic STC mechanism, that is based on the reference function $V_\refer$, by Algorithm~\ref{algo_trig_window} with $C(x,\eta)$ according to \eqref{eq_ref_C}. This leads us to the following result. 
\begin{theo}
	\label{theo_ref}
	Assume there are $n_\tpar$ different parameter sets $\epsilon_i, \gamma_i, L_i$, $i \in \left\lbrace 1,\dots,n_\tpar\right\rbrace$, for which Assumption~1 holds with the same functions $V$ and $\alpha_w$. Let $\epsilon_1>0$ and $\eta(0,0) = V(x(0,0))$. Further, consider $\mathcal{H}_{STC}$ with $S(\eta,x)$ and $\Gamma(x,\eta)$ defined according to \eqref{eq_S_ref} and by Algorithm~\ref{algo_trig_window} with $C(x,\eta)$ according to \eqref{eq_ref_C} and some $\delta \in 
	\left(0,1\right)$. Then $\svar_{j+1}- \svar_j \geq t_{\min} \coloneqq \delta T_{\max}\left(\gamma_1,L_1+\frac{\epsilon_1}{2}\right)$ and $\mathcal{H}_{STC}$ is ISS.
\end{theo}
\begin{proof}
	The proof is given in Appendix~\ref{proof_theo_ref}.
\end{proof}

A numerical example for the mechanisms presented in this section is given in Section~\ref{sec_ex_1}.
\section{Local results for bounded disturbances}
\label{sec_loc}
Up to this point, we have not posed any assumptions on the disturbance signal $w$ other than it is locally integrable. We have derived guarantees on UGAS and ISS for such a setup. 

However, these results require Assumption~\ref{asum_hybrid_lyap} to hold globally, since the disturbance may be arbitrarily large.  This may be restrictive in some situations. Moreover, it may lead to unnecessarily conservative results, as parameters $\epsilon,\gamma$ and $L$ may be chosen in a less conservative manner if only subsets of the state-space need to be considered when verifying Assumption~\ref{asum_hybrid_lyap}.  To overcome these restrictions, we present in this section local results for the case that $w(t) \in \mathcal{W} \coloneqq \left\lbrace w|\abs{w} \leq \bar w \right\rbrace$ for all $t \geq 0$ and some $\bar w >0$, for which a local version of Assumption~\ref{asum_hybrid_lyap} can be exploited.

\subsection{Using disturbance bounds in the general framework for dynamic STC}
Subsequently, we will modify the dynamic STC framework from the previous sections such that RAS of a sublevel set of $V$ will be guaranteed for bounded disturbances.  

To obtain guarantees for RAS, we will use the following local version of Assumption~\ref{asum_hybrid_lyap}, that is stated for a sublevel set $\mathcal{X}_c \coloneqq \left\lbrace x | V(x) \leq c \right\rbrace$ of $V$ for some $c > 0$. We use subsequently the notation $\hat{x} \coloneqq \begin{bmatrix}
	\hat{x}_p^\top&
	\hat{x}_c^\top
\end{bmatrix}^\top$. 
\begin{asum}
	\label{asum_hybrid_lyap_loc}
	 There exist a locally Lipschitz function $W:\mathbb{R}^{n_e} \rightarrow \mathbb{R}_{\geq0}$, a locally Lipschitz function $V:\mathbb{R}^{n_x} \rightarrow \mathbb{R}_{\geq0}$, a continuous function $H:\mathbb{R}^{n_x}\times\mathbb{R}^{n_e}\times \mathcal{W} \rightarrow \mathbb{R}_{\geq0}$, constants $L, \gamma\in \mathbb{R}_{>0}$, $\epsilon\in \mathbb{R}$, and  $\underline{\alpha}_W$, $\overline{\alpha}_W, \underline{\alpha}_V, \overline{\alpha}_V,\alpha_w \in \mathcal{K}_\infty$  such that for all $e \in \mathbb{R}^{n_e}$
	\begin{equation}
		\label{eq_w_bound_loc}
		\underline{\alpha}_W(\abs{e}) \leq W(e) \leq \overline{\alpha}_W(\abs{e}),
	\end{equation}
	for all $x \in \mathbb{R}^{n_x}$,
	\begin{equation}
		\label{eq_V_bound_K_loc}
		\underline{\alpha}_V(\abs{x}) \leq V(x) \leq \overline{\alpha}_V(\abs{x}),
	\end{equation}
	and for all  $x \in \mathcal{X}_c, w \in \mathcal{W} $ and almost all $e = \hat{x} - x$ with $\hat{x} \in \mathcal{X}_c,$ 
	\begin{equation}
		\left\langle \frac{\partial W(e)}{\partial e},g(x,e,w)\right\rangle \leq L W(e) + H(x,e,w).  \label{eq_w_est_loc}
	\end{equation}
	Moreover, for all $e = \hat{x} -x$, with $\hat{x} \in \mathcal{X}_c$, $x \in \mathcal{X}_c$, all  $w \in \mathcal{W}$  and almost all $x \in \mathcal{X}$,
	\begin{equation}
		\begin{split}
			\left\langle \nabla V(x),f(x,e,w) \right\rangle
			\leq &- \epsilon V(x) -H^2(x,e,w)\\ 
			&+ \gamma^2 W^2(e) + \alpha_w\left(\abs{w}\right).
		\end{split}
		\label{eq_v_desc_hybrid_loc}
	\end{equation}
\end{asum}
We will later exploit that the modified assumption may become less restrictive in the sense that smaller values for $\gamma$ and $L$ may be possible as $c$ decreases. Based on the modified assumption, we obtain the following modified version of Proposition~\ref{prop_hybrid}.
\begin{prop}
	\label{prop_hybrid_loc}
	Consider the hybrid system $\mathcal{H}_{STC}$ at sampling instant $\tvar_j^+$ for $j \in \mathbb{N}_0$.  Let Assumption~\ref{asum_hybrid_lyap_loc} hold 
	for $\gamma, \epsilon$, $L$ and $c$ and let $w(t) \in \mathcal{W}$ for all $t \geq 0$.
	Moreover, let $0 < \auxvar(\tvar_j^+) < \delta T_{\max} \left(\gamma,\max\left\lbrace L-\frac{\epsilon}{2},(1-\delta)\right\rbrace \right)$  for some $\delta \in \left(0,1\right)$.	
	Consider $U(\xi)$ according to \eqref{eq_def_u}. 
If 
\begin{equation}
	\label{eq_bound_U_loc}
	\begin{split}
		 & e^{ -\epsilon\auxvar\left(\tvar_j^+\right)}V(x(\tvar_j)) 
		+ \frac{\alpha_w\left(\abs{\bar w}\right)}{\epsilon} \left(1-e^{-\epsilon\auxvar(\tvar_j^+)}\right)\leq c
	\end{split}	
\end{equation}
 and $V(x(\tvar_j)) \leq c$, then  
	\begin{equation}
	\begin{split}
	&V(x(\svar,j+1))\\ 
	\leq& U(\xi(\svar,j+1)) 
	\leq	e^{ -\epsilon (t-\svar_j)}V(x(\tvar_j))\\
	& + \frac{\alpha_w\left(\abs{\bar w}\right)}{\epsilon} \left(1-e^{-\epsilon(t-\svar_j)}\right)
	\end{split}
	\label{eq_prop_hybrid1_loc}
	\end{equation}
	holds for all $\svar_j \leq t \leq \svar_j+\auxvar(\tvar_j^+)$.	
\end{prop}
\begin{proof}
	The proof is given in Appendix~\ref{proof_prop_hyb_loc}.
\end{proof}
Note that we have already used the simplification $\Lambda = \max\left\lbrace L+\frac{\epsilon}{2},(1-\delta)\right\rbrace$ in the proposition. 

Proposition~\ref{prop_hybrid_loc} can be used in order to determine sampling instants such that condition~\eqref{eq_cond_dec} is satisfied for all possible disturbance signals that satisfy the disturbance bound (and not only in the nominal case). Suppose there are $n_\tpar$ parameter sets $\epsilon_\newi,\gamma_\newi,L_\newi,~i\in\left\lbrace1,\dots,n_\tpar\right\rbrace$, for which Assumption~\ref{asum_hybrid_lyap_loc} holds for some $c_\varl$ and with $\epsilon_{1,\varl} > 0$ and $\epsilon_\newi < 0$ for all $i \in \left\lbrace 2,\dots, n_\tpar \right\rbrace$.

If $V(x(\tvar_j)) <c_l, C(x(\tvar_j),\eta(\tvar_j)) \leq c_l$ and $w(t) \in \mathcal{W}$ for $\svar_j \leq t \leq \svar_{j+1}$, then \eqref{eq_cond_dec} can be ensured with Proposition~\ref{prop_hybrid_loc}, if there is a parameter set with index $i_j \in \left\lbrace2,\dots,n_\tpar\right\rbrace$ for Assumption~\ref{asum_hybrid_lyap_loc} for which
\begin{equation}
\label{eq_dec_gen_loc}
\begin{split}
&e^{ -\epsilon_{\newij} (t-\svar_j)}\left(V(x(\tvar_j)) - \frac{\alpha_w\left(\abs{\bar w}\right)}{\epsilon_{\newij}} \right)\\
\leq& e^{-\epsilon_\refer(t - \svar_j)} C(x(\tvar_j),\eta(\tvar_j)) - \frac{\alpha_w\left(\abs{\bar w}\right)}{\epsilon_{\newij}} 
\end{split}
\end{equation} 
holds for $\svar_j \leq t \leq \svar_{j+1}$ and 
\begin{equation}
\label{eq_bound_aux_loc}
\begin{split}
&t-\svar_j \leq \svar_{j+1}-\svar_j = \auxvar(s_j^+)\\
<& \delta T_{\max}\left(\gamma_{\newij},\max\left\lbrace L_{\newij}+\frac{\epsilon_{\newij}}{2}, (1-\delta)\right\rbrace\right)
\end{split}
\end{equation} 
holds for some $\delta > 0$. Note that 
\begin{equation}
	- \frac{\alpha_w\left(\abs{\bar w}\right)}{\epsilon_{\newij}} \geq -e^{-\epsilon_\refer (t-\svar_j)}  \frac{\alpha_w\left(\abs{\bar w}\right)}{\epsilon_{\newij}} 
\end{equation}
holds if $\epsilon_{i_j} \leq 0$. 

A checkable sufficient condition for \eqref{eq_dec_gen_loc}, that can be used to determine $\auxvar(\tvar_j^+)$, can thus be derived for the case that $C(x(\tvar_j),\eta(\tvar_j)) \geq V(x(\tvar_j))$ as 
\begin{equation}
	\begin{split}
		&\left(-\epsilon_{\newij} +\epsilon_\refer\right)(t-\svar_j)
	\leq \log\left(\frac{C(x(\tvar_j),\eta(\tvar_j)) -   \frac{\alpha_w\left(\abs{\bar w}\right)}{\epsilon_{\newij}}}{V(x(\tvar_j))-\frac{\alpha_w\left(\abs{\bar w}\right)}{\epsilon_{\newij}}}\right).
	\end{split}
\label{eq_ln_loc}
\end{equation}

In order to reduce potential conservativity when determining sampling instants, different values for $c$ can be used for Assumption~\ref{asum_hybrid_lyap_loc} at different sampling instances depending on the current values of $V(x(\svar_{j}))$ and $C(x(\svar_{j}),\eta(\svar_{j}))$. In particular, suppose there are $n_\tc \in \mathbb{N}$ variables $c_l, l \in \left\lbrace 1,\dots,n_c\right\rbrace$, for each of which Assumption~\ref{asum_hybrid_lyap_loc} has been verified offline for $n_\tpar$ specific parameter sets $\epsilon_{\newi}, \gamma_\newi$ and $L_\newi$. Then, choosing $l$ such that $c_l$ is as small as possible and satisfies $c_l \geq \max\left\lbrace V(x(\svar_j)),C(x(\svar_{j}),\eta(\svar_{j})) \right\rbrace$, leads to reduced conservativity when determining sampling instants. 

	\begin{algorithm}[tb]
		\caption{Computation of $\Gamma(x,\eta)$ for the dynamic STC framework for some $\delta \in \left(0,1\right)$ and given $C(x,\eta)$, taking into account the disturbance bound.}
		\label{algo_trig_window_loc}
		\begin{algorithmic}[1]
			\STATE $V \leftarrow V(x)$, $C \leftarrow C(x,\eta)$ %
			\STATE $l = \underset{l\in\left\lbrace 1,\dots,n_c \right\rbrace}{\min}~l$ s.t. $c_l \geq \max\left\lbrace V,C \right\rbrace$ \label{line_l}
			\STATE $\bar h \leftarrow \delta T_{\max}\left(\gamma_\itildel,L_\itildel+\frac{\epsilon_\itildel}{2}\right)$ \label{line_fallback_loc}
			\FOR{\text{\bf each} $i \in \left\lbrace2,\dots,n_\tpar\right\rbrace$ } 
			\STATE $\Lambda_{\newi} \leftarrow \max \left\lbrace L_\newi + \frac{\epsilon_\newi}{2},(1-\delta) \right\rbrace$
			\IF{$C \geq V$} \label{line_for_start_loc} %
			\IF{$-\epsilon_\newi+\epsilon_\refer > 0$}
			\STATE $\bar h_i \leftarrow \min\left\lbrace \delta T_{\max}(\gamma_\newi,\Lambda_\newi),\vphantom{\frac{\log\left(C-\frac{\alpha_w(\abs{\bar w}) }{\epsilon_\newi}\right)-\log\left(V-\frac{\alpha_w(\abs{\bar w})}{\epsilon_\newi}\right)}{ -\epsilon_\newi+ \epsilon_\refer} }\right.$
			
			$\hphantom{\bar h_i \leftarrow  }\left.\frac{\log\left(C-\frac{\alpha_w(\abs{\bar w}) }{\epsilon_\newi}\right)-\log\left(V-\frac{\alpha_w(\abs{\bar w})}{\epsilon_\newi}\right)}{ -\epsilon_\newi+ \epsilon_\refer} \right\rbrace$ \label{line_hi_loc}
			\ELSE
			\STATE $\bar h_i \leftarrow \delta T_{\max}(\gamma_\newi,\Lambda_\newi)$
			\ENDIF
			\ELSE
			\STATE $\bar h_i \leftarrow 0$
			\ENDIF\label{line_for_end_loc}
			\IF{$\bar h_i > \bar h$}
			\STATE $\bar h \leftarrow \bar h_i$\label{line_h_update_loc}
			\ENDIF
			\ENDFOR 			
			\STATE $\Gamma(x,\eta) \leftarrow \bar{h}$
		\end{algorithmic}
	\end{algorithm}

A modified version of Algorithm~\ref{algo_trig_window} that considers explicitly the bound on $w$ and that uses different values for $c$ in Assumption~\ref{asum_hybrid_lyap_loc} is given by Algorithm~\ref{algo_trig_window_loc}. Note that this algorithm follows essentially the same main steps as Algorithm~\ref{algo_trig_window}. 

Differences to Algorithm~\ref{algo_trig_window} are that first a suitable value for $l$ is chosen in Line~\ref{line_l} such that $c_l$ is as small as possible but satisfies $c_l \geq \max \left\lbrace V(x(\tvar_j)),C(x(\tvar_j),\eta(\tvar_j)) \right\rbrace$. Then, an iteration over all parameter sets is started. For each corresponding parameter set, a preferably large value for $\bar h_i$ is determined such that \eqref{eq_cond_dec} can be guaranteed if $\svar_{j+1} = \svar_j + \bar h_i$  for $l$ and for all possible disturbance realizations. This is ensured by the choice of $\bar{h}_i$ in Line~\ref{line_hi_loc} that is modified in comparison to the respective line in Algorithm~\ref{algo_trig_window_loc}. Similar as in Algorithm~\ref{algo_trig_window}, the maximum such $\bar h_i$ is selected as value for $\Gamma(x(\tvar_j),\eta(\tvar_j)) = \bar h$. If there is no parameter set for the considered $l$, for which $\eqref{eq_cond_dec}$ can be guaranteed to hold based on Proposition~\ref{proof_prop_fir_loc}, the fall-back strategy based on $\epsilon_{1,l}$ is used. In particular, the Algorithm sets in this case $\bar h = \delta T_{\max}(\gamma_\itildel,L_\itildel+\frac{\epsilon_\itildel}{2})$. For this fall-back strategy, it is important to note that Proposition~\ref{prop_hybrid_loc} delivers in this case 
 that 
\begin{equation}
\label{eq_eps1_loc}
\begin{split}
V(x(\tvar_{j+1})) 
\leq&	e^{ -\epsilon_1 (\svar_{j+1}-\svar_j)}V(x(\tvar_j))\\
& + \frac{\alpha_w\left(\abs{\bar w}\right)}{\epsilon_{1,l}} \left(1-e^{-\epsilon_{1,l}(\svar_{j+1}-\svar_j)}\right).
\end{split}
\end{equation} 
To guarantee RAS of the set $\mathcal{R}$ with ROA $\mathcal{X}_{c_{\max}}$, $c_w$ needs to be such that $V(x(\tvar_{j+1})) < V(x(\tvar_j))$ holds for the fall-back strategy for $x(\tvar_j) \in \mathcal{X}\backslash\mathcal{R}$ and $V(x(\tvar_{j+1})) \leq c_w$ for all $x(\tvar_j) \in \mathcal{R}$. This is ensured by \eqref{eq_eps1_loc} if $c_w \geq \underset{l \in \left\lbrace 1,\dots,n_c \right\rbrace}{\max} \frac{\alpha_w\left(\abs{\bar w}\right)}{\epsilon_{1,l}}$, which we use as a lower bound for possible values of $c_w$. It thus also determines the minimum size of $\mathcal{R}$. Moreover, it is necessary that $c_w \leq
 c_{\max} \leq \underset{l \in \left\lbrace 1,\dots,n_c \right\rbrace}{\max} c_l$ such that suitable parameter sets for Assumption~\ref{asum_hybrid_lyap_loc} are available for all $x \in \mathcal{X}_{c_{\max}}$.

Next, we discuss how the particular dynamic STC mechanisms from Section~\ref{sec_spec} needs to be modified in order to guarantee RAS of the set $\mathcal{R}$ with ROA $\mathcal{X}_{c_{\max}}$.  

\subsection{Modifications for the dynamic STC mechanisms to guarantee RAS}
We discuss in this subsection, which additional modifications are required for the particular dynamic STC mechanisms from Section~\ref{sec_spec} in order to guarantee RAS of the set $\mathcal{R}$ with region of attraction $\mathcal{X}_{c_{\max}}$. For simplicity, we assume that $c_{\max} = \underset{l \in \left\lbrace 1,\dots,n_c \right\rbrace}{\max} c_l$ and $c_w = \underset{l \in \left\lbrace 1,\dots,n_c \right\rbrace}{\min} c_l$. 
\subsubsection{Dynamic STC based on an FIR filter} 
To guarantee RAS  of the set $\mathcal{R}$ with region of attraction $\mathcal{X}_{c_{\max}}$, the dynamic STC mechanism based on an FIR filter for $V$ from Subsection~\ref{subsec_fir} needs to be modified such that it ensures that the system state stays in the set $\mathcal{X}_{c_{\max}}$ for all times, since Assumption~\ref{asum_hybrid_lyap_loc} is only valid in this set. Keeping the system state in $\mathcal{X}_{c_{\max}}$ can be achieved by limiting the maximum value of $C(x,\eta)$ by $c_{\max}$. Moreover, in order to enlarge the time between sampling instants, it is beneficial to set the value of $C(x,\eta)$ to $c_w$ if the filter state would else result in smaller values. Both can be achieved by replacing the definition of $C(x,\eta)$ from \eqref{eq_fir_C} by 
\begin{equation}
\label{eq_fir_C_loc}
C(x,\eta) = \max\left\lbrace\min\left\lbrace\frac{1}{m} \left(V(x)+ \sum_{k=1}^{m-1} \eta_k\right),c_{\max}\right\rbrace,c_w\right\rbrace.
\end{equation}
Using this modification, we obtain the following result.
\begin{theo}
	\label{prop_fir_loc}
	Suppose  $w(t) \in \mathcal{W}$ for all $t\geq 0$. Assume there are $n_c$ parameters $c_l, l \in \left\lbrace 1,\dots,n_c\right\rbrace$, each with $n_\tpar$ different parameter sets $\epsilon_\newi, \gamma_\newi, L_\newi$, $i \in \left\lbrace1,\dots,n_\tpar\right\rbrace$, for which Assumption~2 holds
	 with the same functions $V$ and $\alpha_w$.  Let $\epsilon_{1,l}>0$ and $\epsilon_{\newi} < 0$ for $i > 1$ and each $l$. Further, let $c_{\max} \geq c_w \geq \underset{l \in \left\lbrace 1,\dots,n_c\right\rbrace}{\max} \frac{\alpha_w\left(\abs{\bar w}\right)}{\epsilon_{1,l}}$. Consider $\mathcal{H}_{STC}$ with $S(\eta,x)$ and $\Gamma(x,\eta)$ defined according to \eqref{eq_S_window} and by Algorithm~\ref{algo_trig_window_loc} with $C(x,\eta)$ according to \eqref{eq_fir_C_loc} and some $\delta \in 
	\left(0,1\right)$. Then $\svar_{j+1}- \svar_j \geq t_{\min} = \underset{l \in \left\lbrace 1,\dots, n_c \right\rbrace}{\min}\delta T_{\max}\left(\gamma_{1,l},L_{1,l}+\frac{\epsilon_{1,l}}{2}\right)$ and the set $\mathcal{R}$ is RAS for $\mathcal{H}_{STC}$ with region of attraction $\mathcal{X}_{c_{\max}}$.
\end{theo}
\begin{proof}
	The proof is given in Appendix~\ref{proof_prop_fir_loc}.
\end{proof}

\subsubsection{Dynamic STC based on an IIR filter}
The modification required for the dynamic STC mechanism based on an IIR filter for $V$ from Subsection~\ref{subsec_iir} to guarantee RAS of $\mathcal{R}$ for the ROA $\mathcal{X}_{c_{\max}}$ is quite similar as for the FIR mechanism. To ensure that the system state stays in $\mathcal{X}_{c_{\max}}$ for all times, the maximum value of $C(x,\eta)$ can again be limited by $c_{\max}$. Moreover, similar as the modification for the FIR mechanism, it is beneficial to set the value of $C(x,\eta)$ to $c_w$ if the filter state would result in smaller values in order to enlarge the time between sampling instants. Both can be achieved by replacing the definition of $C(x,\eta)$ from \eqref{eq_iir_C} by
\begin{equation}
	\label{eq_iir_C_loc}
	C(x,\eta) = \max\left\lbrace \min\left\lbrace \eta ,c_w \right\rbrace, c_{\max} \right\rbrace.
\end{equation} 
With this modification, we obtain the following result.
\begin{theo}
	\label{theo_iir_loc}
		Suppose  $w(t) \in \mathcal{W}$ for all $t\geq 0$. Assume there are $n_c$ parameters $c_l, l \in \left\lbrace 1,\dots,n_c\right\rbrace$ each with $n_\tpar$ different parameter sets $\epsilon_\newi, \gamma_\newi, L_\newi$, $i \in \left\lbrace1,\dots,n_\tpar\right\rbrace$, for which Assumption~2 holds
		with the same functions $V$ and $\alpha_w$.  Let $\epsilon_{1,l}>0$ and $\epsilon_{\newi} < 0$ for $i > 1$ and each $l$. Further, let $c_{\max} \geq c_w \geq \underset{l \in \left\lbrace 1,\dots,n_c\right\rbrace}{\max} \frac{\alpha_w\left(\abs{\bar w}\right)}{\epsilon_{1,l}}$, $r_1,r_2 > 0$ and $r_1+r_2 \leq 1$.
	 Consider $\mathcal{H}_{STC}$ with $S(\eta,x)$ and $\Gamma(x,\eta)$ defined according to \eqref{eq_S_iir} and by Algorithm~\ref{algo_trig_window_loc}, $C(x,\eta)$ according to \eqref{eq_iir_C_loc} and some $\delta \in 
	\left(0,1\right)$. Then $\svar_{j+1}- \svar_j \geq t_{\min} = \underset{l \in \left\lbrace 1,\dots, n_c \right\rbrace}{\min}\delta T_{\max}\left(\gamma_{1,l},L_{1,l}+\frac{\epsilon_{1,l}}{2}\right)$ and the set $\mathcal{R}$ is RAS for $\mathcal{H}_{STC}$ with region of attraction $\mathcal{X}_{c_{\max}}$.
\end{theo}
The proof of Theorem~\ref{theo_iir_loc} is omitted for brevity. It follows from a combination of the proofs of Theorems~\ref{theo_iir} and \ref{prop_fir_loc}.
\subsubsection{Dynamic STC based on a reference function}
To modify the reference function-based dynamic STC mechanism from Subsection~\ref{subsec_ref} to ensure RAS of $\mathcal{R}$ with ROA $\mathcal{X}_{c_{\max}}$, the reference function needs to be adapted. To ensure that the system state stays in $\mathcal{X}_{c_{\max}}$ for all times, the maximum value of the reference function needs to be bounded. Moreover, instead of choosing the reference function such that it converges to $0$, it is beneficial to let it converge to $c_w$ in order to maximize the time between sampling instants. We focus subsequently on the specific reference function
\begin{equation*}
	\begin{split}
		&V_\refer(t,x(0,0)) \\
		\coloneqq& c_w + e^{-\epsilon_\refer t} \min\left\lbrace\max\left\lbrace\left(V(x(0,0)) - c_w\right),0 \right\rbrace,c_{\max}-c_w\right\rbrace,
	\end{split}
\end{equation*}
which includes both modifications. We assume that $\eta(0,0) = \min\left\lbrace\max\left\lbrace\left(V(x(0,0)) - c_w\right),0 \right\rbrace,c_{\max}-c_w\right\rbrace$. Then, the reference function can be implemented with $n_\eta = 1$ and \eqref{eq_S_ref} where $\Gamma(x,\eta)$ is defined by  Algorithm~\ref{algo_trig_window_loc} with 
\begin{equation}
	\label{eq_ref_C_loc}
	C(x,\eta) = c_w+\eta.
\end{equation}
Note that for these choices, $\Gamma(x,\eta)$ is chosen if possible such that \eqref{eq_ref_goal} holds not only in the nominal case, but for all disturbance signals that may occur. We obtain the following result for the modified mechanism. 
\begin{theo}
	\label{theo_ref_loc}
	Suppose  $w(t) \in \mathcal{W}$ for all $t\geq 0$. Assume there are $n_c$ parameters $c_l, l \in \left\lbrace 1,\dots,n_c\right\rbrace$, each with $n_\tpar$ different parameter sets $\epsilon_\newi, \gamma_\newi, L_\newi$, $i \in \left\lbrace1,\dots,n_\tpar\right\rbrace$, for which Assumption~2 holds
	with the same functions $V$ and $\alpha_w$.  Let $\epsilon_{1,l}>0$ and $\epsilon_{\newi} < 0$ for $i > 1$ and each $l$. Further, let $c_{\max} \geq c_w \geq \underset{l \in \left\lbrace 1,\dots,n_c\right\rbrace}{\max} \frac{\alpha_w\left(\abs{\bar w}\right)}{\epsilon_{1,l}}$ and $\eta(0,0) = \min\left\lbrace\max\left\lbrace\left(V(x(0,0)) - c_w\right),0 \right\rbrace,c_{\max}-c_w\right\rbrace$. Consider $\mathcal{H}_{STC}$ with $S(\eta,x)$ and $\Gamma(x,\eta)$ defined according to \eqref{eq_S_ref} and by Algorithm~\ref{algo_trig_window_loc}, $C(x,\eta)$ according to \eqref{eq_ref_C_loc} and some $\delta \in 
	\left(0,1\right)$. Then $\svar_{j+1}- \svar_j \geq t_{\min} = \delta T_{\max}\left(\gamma_{1,l},L_{1,l}+\frac{\epsilon_{1,l}}{2}\right)$ and the set $\mathcal{R}$ is RAS for $\mathcal{H}_{STC}$ with region of attraction $\mathcal{X}_{c_{\max}}$.
\end{theo}
The proof of Theorem~\ref{theo_ref_loc} is omitted for brevity. It follows from a combination  of the proofs of Theorems~\ref{theo_ref} and \ref{prop_fir_loc}.
\section{Numerical Example}
\label{sec_ex}
In this section, we present two numerical examples to illustrate the dynamic STC mechanisms from Sections~\ref{sec_spec}~and~\ref{sec_loc}. 
\subsection{Example~1}
\label{sec_ex_1}
In this subsection, we illustrate the dynamic STC mechanisms from Section~\ref{sec_spec} with the nonlinear example system from \cite{postoyan2014tracking}. The example system is a perturbed single-link robot arm described by
\begin{equation}
	\begin{split}
		\dot x_1 =& x_2+w\\
		\dot x_2 =& -a \sin(x_1) + b \hat{u}
	\end{split}
\end{equation} 
with the static state feedback controller 
	$u = b^{-1}\left(a \sin({x}_1) - {x}_1 -{x}_2\right).$
We define $x = \begin{bmatrix}
	x_1 & x_2
\end{bmatrix}^\top$ and $e = \begin{bmatrix}
e_1 & e_2
\end{bmatrix}^\top = \begin{bmatrix}
	\hat{x}_1 - x_1 & \hat{x}_2 - x_2
\end{bmatrix}^\top$. Observe that 
\begin{equation*}
	\begin{split}
		&-a\left(\sin(x_1) - \sin(x_1+e_1)\right)\\
		=&2\cos\left(\frac{2x_1+e_1}{2} \sin\left(-\frac{e_1}{2}\right)\right) = \tilde{a} e_1
	\end{split}
\end{equation*}
for a varying parameter $\tilde{a} \in \left[-a,a\right]$, depending on $x_1$. Hence, we obtain 
\begin{equation*}
	\begin{split}
		f(x,e,w) =& \begin{bmatrix}
			x_1 & -x_1-x_2+(\tilde{a} -1)e_1-e_2+w
		\end{bmatrix}^\top\\
	&=\begin{bmatrix}
		0 & 1\\
		-1 & -1
	\end{bmatrix} x 
	+\begin{bmatrix}
		0 & 0\\
		\tilde{a}-1 & 1
	\end{bmatrix} e+ \begin{bmatrix}
	0\\ 1
\end{bmatrix} w.
	\end{split}
\end{equation*}
For $V(x) = x^\top P x$, $\alpha_w(\abs{w}) = \theta^2 w^2$ for $\theta > 0$ and any fixed $\tilde{a}\in\left[-a, a\right]$, the LMI-based approach from \cite[Section~4]{hertneck20simple} to verify Assumption~\ref{asum_hybrid_lyap} can be easily adapted to the setup from this article. 
In particular, \eqref{eq_w_est} is convex in $\tilde{a}$ and can thus be verified for all $\tilde{a}\in\left[-a,a\right]$ by taking the maximum value of $L$ for the extremal values $\tilde{a} = -a$ and $\tilde{a} = a$. Inequality~\eqref{eq_v_desc_hybrid} can be factorized such that the result is convex in $\tilde{a}$ and $\tilde{a}^2$. Then, $\gamma$ can be minimized with one LMI constraint similar as in \cite[Section~4]{hertneck20simple} for each combination of the extremal values for $\tilde{a}$ and $\tilde{a}^2$.

Subsequently, we consider $a = \frac{9.81}{2}$ and $b = 2$. We have computed $n_\tpar = 21$ different parameter sets that satisfy Assumption~\ref{asum_hybrid_lyap_loc} with $\epsilon_i \in \left[-20,0.01\right]$ and $\theta = 10$. The maximum sampling interval, for which ISS can be guaranteed for periodic sampling, is \SI{0.175}{\second}. It serves also as fall-back strategy for the dynamic STC mechanisms. 
\begin{figure}
	\centering
	\resizebox{.92\linewidth}{!}{
		\input{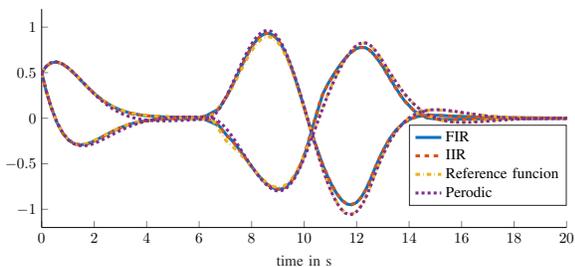}
	}
	\caption{State trajectories for the dynamic STC mechanisms from Section~\ref{sec_spec} and periodic sampling for a simulation with $x(0) = \left[0.5,0.5\right]^\top$.
	}
	\label{fig_states}
\end{figure}
\begin{figure}
	\centering
	\resizebox{.92\linewidth}{!}{
		\input{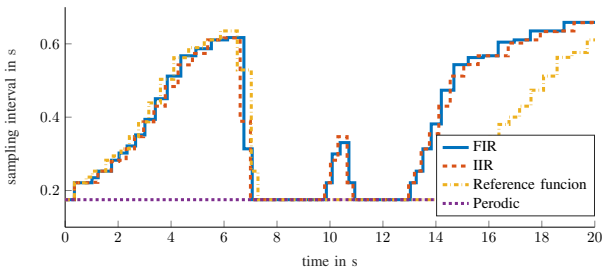}}
	\caption{Comparison of the sampling intervals for dynamic STC mechanisms from Section~\ref{sec_spec} and periodic sampling for a simulation with $x(0) = \left[0.5,0.5\right]^\top$.
	}
	\label{fig_inters}
\end{figure}

In Figures~\ref{fig_states} and \ref{fig_inters}, state trajectories and the evolution of the sampling intervals for simulations of the three dynamic STC mechanisms from Section~\ref{sec_spec}  and of periodic sampling with sampling period  \SI{0.175}{\second} are plotted. The initial condition for the trajectories is $x(0) = \left[0.5,0.5\right]^\top$. The disturbance signal is $w(t) = \sin(t)$ for $ 2\pi\SI{}{\second}\leq t \leq  4\pi\SI{}{\second}$ and $w(t) = 0$ else.  For the three dynamic STC mechanisms, we have used $\epsilon_\refer = 0.2$. For the FIR mechanism from Subsection~\ref{subsec_fir}, we have in addition chosen $n_\eta = 20$ and for the IIR mechanism from Subsection~\ref{subsec_iir}, we have chosen $r_1 = 0.9$ and $r_2 = 0.1$. 

It can be seen that significantly larger sampling intervals are achieved by all dynamic STC mechanisms in comparison to periodic sampling. Nevertheless, the state trajectories are qualitatively similar. All three dynamic STC mechanisms reduce the sampling intervals as soon as the disturbance drives the system states away from the origin. When the influence of the disturbance reduces, the sampling intervals are increased again. When comparing the three dynamic STC mechanisms, the FIR and IIR mechanisms show similar behavior. The reference function mechanism takes longer to enlarge the sampling intervals after the influence of the disturbance on the system state has diminished.

Note that all the dynamic STC mechanisms from Section~\ref{sec_spec} aim for stabilizing the origin, which may be disadvantageous when this is impossible due to the disturbance, since then the sampling interval may even be reduced to the fall-back strategy. This potential disadvantage is overcome by aiming to stabilize an invariant set containing the origin, as it is done by the dynamic STC mechanisms from Section~\ref{sec_loc}.
\subsection{Example~2}
\color{black}
In this subsection, we use the example from \cite{tiberi2013simple} to illustrate the dynamic STC mechanisms from Section~\ref{sec_loc}. The example system is given by
\begin{equation*}
	\begin{split}
		\dot{x}_1 &= -x_1\\
		\dot{x}_2 &= \left(x_1^2+x_2^2\right)x_2+\hat{u}+w
	\end{split}
\end{equation*}
with the static state feedback controller
	$u = -(1+x_1^2+x_2^2)x_2.$
We assume that $\abs{w(t)} \leq 0.4$ for all $t$. Using again $x = \begin{bmatrix}
	x_1 & x_2
\end{bmatrix}^\top$ and $e = \begin{bmatrix}
	e_1 & e_2
\end{bmatrix}^\top = \begin{bmatrix}
	\hat{x}_1 - x_1 & \hat{x}_2 - x_2
\end{bmatrix}^\top$, we obtain
\begin{equation}
	\label{eq_f_ex_2}
	\begin{split}	
	&f(x,e,w)\\
	 =& \begin{bmatrix}
			-1& 0\\
			0& -1
		\end{bmatrix} x 
	+ \begin{bmatrix}
			0 & 0\\
			a_1(x,e) & a_2(x,e)-1 
		\end{bmatrix} e 
		+\begin{bmatrix}
			0\\1
		\end{bmatrix}w
	\end{split}
\end{equation}
with $a_1(x,e) = -2x_1x_2-x_2e_1$ and $a_2(x,e) = -2x_2^2-x_2e_2-\hat{x}_1^2-\hat{x}_2^2$.
We consider $V(x) = 1.5x^\top x$. Note that for any $c_l \in \mathbb{R}$, all $x\in\mathcal{X}_{c_{l}}$ satisfy $\abs{x} \leq \sqrt{\frac{c_l}{28}}$ and all $e = \hat{x}-x$ with $\hat{x}\in\mathcal{X}_{c_{l}}$ and $x\in\mathcal{X}_{c_{l}}$ satisfy $\abs{e} \leq 2 \sqrt{\frac{c_l}{28}}$. Thus \eqref{eq_f_ex_2} can be rewritten 
as
$
	f(x,e,w) = \begin{bmatrix}
		-1& 0\\
		0& -1
	\end{bmatrix} x 
	+ \begin{bmatrix}
		0 & 0\\
		\tilde{a}_1 & \tilde{a}_2 
	\end{bmatrix} e 
	+\begin{bmatrix}
		0\\1
	\end{bmatrix}w
$
for varying parameters $\tilde{a}_1 \in \left[-\frac{c_l}{7},\frac{c_l}{7}\right]$ and $\tilde{a}_2 \in \left[-\frac{3c_l}{14},\frac{3c_l}{14}\right]$, depending on $x$ and $e$. Thus, Assumption~\ref{asum_hybrid_lyap_loc} can be verified for this example for any $c_l$ and any $x\in\mathcal{X}_{c_l}$ and $\hat{x}\in\mathcal{X}_{c_l}$ by using the LMI-based approach from \cite[Section~4]{hertneck20simple} for all combinations of extremal values of $\tilde{a}_1$ and $\tilde{a}_2$.

 Similar as in \cite{tiberi2013simple}, we aim for stabilizing the set\footnote{Note that the ultimate bound in \cite{tiberi2013simple} is stated wrongly to be $0.325$.} $\left\lbrace x|\abs{x} \leq 0.65 \right\rbrace$ for all initial conditions $x(0)$ that satisfy $\abs{x(0)} \leq 5$. For our choice of $V(x)$, this translates to $c_w = 0.64$ and $c_{\max} = 37.87$ for $\mathcal{R}$ and $\mathcal{X}_{c_{\max}}$. 
We have selected $n_l = 40$ variables $c_l, l\in\left\lbrace 1,\dots,40\right\rbrace$ and  for each $c_l$ computed $n_\tpar = 20$ different parameter sets that satisfy Assumption~\ref{asum_hybrid_lyap_loc} with $\epsilon_{i,l} \in \left[-15,1\right]$ and $\theta = 2$. 

\begin{figure}
	\centering
	\resizebox{.92\linewidth}{!}{
		\input{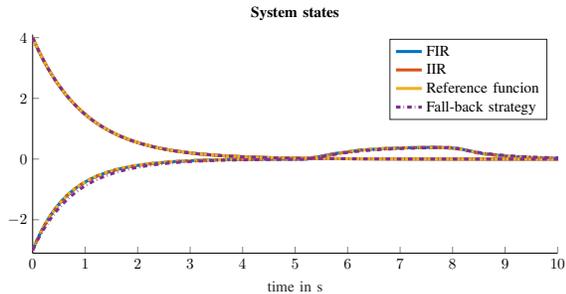}
	}
	\caption{State trajectories for the dynamic STC mechanisms from Section~\ref{sec_loc} and periodic sampling for a simulation with $x(0) = \left[4,-3\right]^\top$.
	}
	\label{fig_states_loc}
\end{figure}
\begin{figure}
	\centering
	\resizebox{.92\linewidth}{!}{
		\input{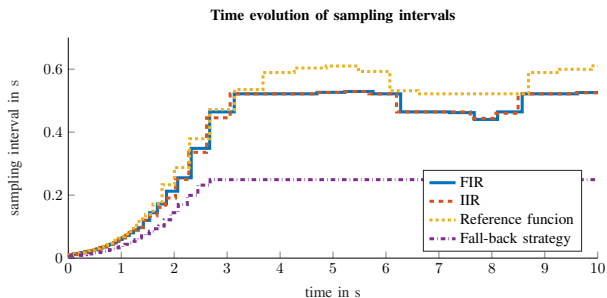}}
	\caption{Comparison of the sampling intervals for dynamic STC mechanisms from Section~\ref{sec_loc} and periodic sampling for a simulation with $x(0) = \left[4,-3\right]^\top$.
	}
	\label{fig_inters_loc}
\end{figure} 

In Figures~\ref{fig_states_loc} and \ref{fig_inters_loc}, state trajectories and the evolution of the sampling intervals for simulations of the three dynamic STC mechanisms from Section~\ref{sec_spec}  and of using always the fall-back strategy as sampling interval for the current sublevel set of $V(x)$, i.e., the largest sampling period for which we could guarantee stability for periodic sampling for this sublevel set, are plotted. The initial condition for the trajectories is $x(0) = \left[4,-3\right]^\top$. The disturbance signal is $w(t) = 0.4$ for $\SI{5.3}{\second}\leq t \leq\SI{8}{\second}$ and $w(t) = 0$ else.  For the three dynamic STC mechanisms, we have used $\epsilon_\refer = 1$. For the FIR mechanism from Subsection~\ref{subsec_fir}, we have in addition chosen $n_\eta = 20$ and for the IIR mechanism from Subsection~\ref{subsec_iir}, we have chosen $r_1 = 0.9$ and $r_2 = 0.1$. 
It can be seen that significantly larger sampling intervals can be achieved for all three dynamic STC mechanisms in comparison to using the fall-back strategy. Moreover, the disturbance does not force the mechanisms to use the sampling interval of the fall-back strategy, since the mechanisms do not longer aim to stabilize the origin, which would be impossible due to the disturbance. 

In Table~\ref{tab_comp}, a comparison of the required number of sampling instants for the three dynamic STC mechanisms from Section~\ref{sec_loc} is given. It can be seen that all three mechanisms lead to a comparable number of sampling instants. When comparing the dynamic STC mechanisms to the static STC mechanism from \cite{tiberi2013simple} it can be seen, that the dynamic STC mechanisms require significantly less sampling instants.

\begin{table}
	\caption{Number of sampling instants for a simulation with $x_0 = \left[4,-3\right]^\top$.}
	\centering
	\begin{tabular}{ |c|c|c|c| } 
		\hline
		FIR & IIR & Reference Function& \cite{tiberi2013simple} \\\hline 
		73 & 73 & 71 & 12907 \\ 
		\hline
	\end{tabular}
	\label{tab_comp}
\end{table} 

Note that an additional comparison of a preliminary version of the dynamic STC mechanism based on the FIR filter to the static STC mechanism from \cite{delimpaltadakis2020isochronous} can be found in \cite{hertneck21dynamic}.

\section{Conclusion}
\label{sec_conc}
This article showed how information about the past system behavior can be exploited to increase sampling intervals for nonlinear self-triggered control. We presented a general framework to encode this past behavior in a dynamic variable.
The general framework allowed us to design different particular STC mechanisms and to study ISS and RAS of the resulting systems using hybrid Lyapunov function techniques.
The ISS variant of the framework has the advantage that no knowledge of the disturbance signal is required. 
 If a bound on the disturbance is known, then additional benefits can be obtained using the RAS variant of the framework. For this variant, the main assumption needs then to hold only locally and less frequent triggering may be possible, since a set with size depending on the disturbance bound is stabilized. Moreover for the RAS variant, the parameters of the STC mechanism can be adapted online depending on the actual sublevel set 
 which the system state is located in. Both variants were extensively studied in numerical examples.
There are still some open points for future research. Currently, information of the entire plant state is required for the STC framework, which may be restrictive. Therefore, modifying the framework to support output feedback, e.g., by using an observer for the plant state, would be beneficial. Moreover, in many NCS setups, sensors are spatially distributed and only one sensor can transmit at a time. Extending the framework to such a setup, e.g., by considering a transmission protocol, would make it applicable to a wider range of NCS. Finally, it would be interesting to extend existing static nonlinear STC approaches from the literature, such as those from \cite{tiberi2013simple,delimpaltadakis2020isochronous} such that they incorporate past information of the plant state
 as well.

	\bibliographystyle{IEEEtran}   
	\bibliography{../../../Literatur/literature}
	
	\appendix
	\subsection{Proofs of main results}
\subsubsection{Proof of Proposition~\ref{prop_hybrid}}
\label{proof_prop_hyb}
	This proof follows the same lines as the proof of \cite[Proposition~12]{hertneck20stability}, but is adapted to the setup of this article.
	Recall from \cite{carnevale2007lyapunov} that $\phi(\tau) \in \left[\lambda, \lambda^{-1}\right]$ for all $\tau \in \left[0, \tilde{T}_{\max} \right]$, where $\tilde{T}_{\max} = \tilde{T}_{\max}(\lambda,\gamma,\lvar )$ with
	\begin{equation*}
	\tilde{T}_{\max}(\lambda,\gamma,\lvar ) \coloneqq \begin{cases}\vspace{1mm}
	\frac{1}{\lvar r} \mathrm{arctan}\left(\frac{r(1-\lambda)}{2 \frac{\lambda}{1+\lambda} \left(\frac{\gamma}{\lvar }-1\right)+1+\lambda}\right) & \gamma > \lvar \\ \vspace{1mm}
	\frac{1}{\lvar } \frac{1-\lambda}{1+\lambda} & \gamma = \lvar \\
	\frac{1}{\lvar r} \mathrm{arctanh}\left(\frac{r(1-\lambda)}{2 \frac{\lambda}{1+\lambda} \left(\frac{\gamma}{\lvar }-1\right)+1+\lambda}\right) &\gamma < \lvar 
	\end{cases}
	\end{equation*}
	and  $r$ defined by \eqref{eq_def_r}.
	For any $\auxvar(s_j^+) < T_{\max} (\gamma,\lvar )$, there is a $\lambda \in \left(0,1\right)$ such that $\auxvar(s_j^+)  = \tilde{T}_{\max}(\lambda,\gamma,\lvar )$ (cf. \cite{nesic2009explicit}).
	We observe with Assumption~\ref{asum_hybrid_lyap} that for $\tau \in \left[0,\auxvar(s_j^+)\right]$, all $w$ and almost all
	 $(x,e)$
	\begin{equation*}
	\begin{split}
	&\left\langle \nabla U(\xi), F(\xi,w) \right\rangle\\
	\leq&-\epsilon V(x) -H^2(x) + \gamma^2 W^2(e) + \alpha_w\left(\abs{w}\right)  \\
	&+2\gamma\phi(\tau)W(e)(L W(e) + H(x))\\
	& - \gamma W^2(e) (2\lvar \phi(\tau)+\gamma(\phi^2(\tau)+1))\\
	\leq&-\epsilon V(x) -(H(x) -\gamma\phi(\tau)W(e))^2\\ 
	&+ 2\gamma\phi(\tau)W^2(e)(L-\lvar )+ \alpha_w\left(\abs{w}\right)\\
	\leq& -\epsilon V(x) + 2\gamma\phi(\tau)W^2(e)(L-\lvar )+ \alpha_w\left(\abs{w}\right).
	\end{split}
	\end{equation*} 
	Thus, we obtain %
	\begin{equation*}
	\label{eq_dec_u}
	\begin{split}
	&\frac{d}{dt} U(\xi(\svar,j+1))\\ 
	\leq & \max \left\lbrace -\epsilon , 2(L-\lvar) \right\rbrace U(\xi(t,j+1))+\alpha_w\left(\abs{w}\right) 
	\end{split}	
	\end{equation*} 
	and hence, due to the comparison Lemma (cf. \cite[p. 102]{khalil2002nonlinear}) and with $U(\xi(s_j^+)) = V(x(s_j^+))=V(x(s_j))$ and $V(x(t,j+1)) \leq U(\xi(t,j))$ since $\phi(\tau) > 0$ for all $\tau \in\left[0,T_{\max}(\gamma,\Lambda)\right)$, we obtain for $\svar_j \leq t \leq \svar_j + \auxvar(\tvar_j^+)$
	\begin{align*}
	 V(x(t,j+1)) \leq&U(\xi(\svar,j+1))\\ \nonumber 
	\leq& e^{ \max \left\lbrace -\epsilon, 2(L - \lvar ) \right\rbrace (t-\svar_j)}U(\xi(\tvar_j^+))\nonumber \\
	&+ \int_{t_j}^{t} e^{\max \left\lbrace -\epsilon, 2(L - \lvar ) \right\rbrace (t-\tau)}
	\alpha_w\left(\abs{w}\right) d\tau \nonumber\\
	\leq& e^{\max \left\lbrace -\epsilon, 2(L - \lvar ) \right\rbrace (t-\svar_j)}V(x(\tvar_j))\nonumber \\
	&+ \int_{t_j}^{t} e^{\max \left\lbrace -\epsilon, 2(L - \lvar ) \right\rbrace (t-\tau)}\alpha_w\left(\abs{w}\right)   d\tau.  \hfill\hfill \qed\nonumber
	\end{align*}

\subsubsection{Proof of Theorem~\ref{prop_fir}}
\label{proof_prop_fir}
	Recall that jumps of $\mathcal{H}_{STC}$ occur exactly at sampling instants that are described by $\tvar_j \coloneqq (\svar_j,j)$ and by $\tvar_j^+ \coloneqq (\svar_j,j+1)$. Therefore, it holds that $e(\tvar_j^+) = 0, \eta(\tvar_j^+) = S(\eta(\tvar_j),x(\tvar_j))$ with $S(\eta,x)$ defined by \eqref{eq_S_window} and that $\auxvar(\tvar_j^+)$ is the output $\Gamma(x(\tvar_j),\eta(\tvar_j))$ of Algorithm~\ref{algo_trig_window} for $C(x(\tvar_j),\eta(\tvar_j))$ defined by \eqref{eq_fir_C}. 
	
	Obviously, $\bar h \geq  \delta T_{\max}\left(\gamma_1,L_1+\frac{\epsilon_1}{2}\right) = t_{\min}$ in Algorithm~\ref{algo_trig_window} and thus, $\svar_{j+1} - \svar_j = \Gamma(x(\tvar_j),\eta(\tvar_j)) \geq t_{\min}$ holds for any $j \in \mathbb{N}_0$. Moreover,  $\svar_{j+1} - \svar_j \leq t_{\max} \coloneqq \underset{i \in \left\lbrace 1,\dots,n_\tpar\right\rbrace}{\max} \delta T_{\max}(\gamma_i,L_i+\frac{\epsilon_i}{2})
	$ holds for any $j \in \mathbb{N}_0$ when Algorithm~\ref{algo_trig_window} terminates due to the updated of $\bar h$ in the algorithm. 
	
	Thus, there is for each $j \in \mathbb{N}_0$ an $\ivar{j} \in \left\lbrace 1,\dots,n_\tpar\right\rbrace$, such that $\svar_{j+1} - \svar_j \leq \delta T_{\max}\left(\gamma_\ivar{j},L_\ivar{j}+\frac{\epsilon_\ivar{j}}{2}\right)$.  Proposition~\ref{prop_hybrid} thus implies
	for $\svar_j \leq t \leq \svar_{j+1}$ with $\Lambda_\ivar{j} = L_\ivar{j}+\frac{\epsilon_\ivar{j}}{2}$ that 
	\begin{equation}
	\label{eq_V_dec_fir}
	\begin{split}
	&V(x(\svar,j+1)) \leq U_\ivar{j}(\xi(t,j+1))	\\
	\leq 	&	e^{ -\epsilon_\ivar{j} (t-\svar_j)}V(x(\tvar_j)) + \int_{\svar_j}^{t} e^{ -\epsilon_\ivar{j} (t-\tau)} \alpha_w\left(\abs{w(\tau)}\right) d\tau\\
	\leq& e^{ -\epsilon_\ivar{j} (t-\svar_j)}V(x(\tvar_j)) 
	+ \max \left\lbrace e^{\epsilon_\refer t_{\max}},e^{\left(-\epsilon_i+\epsilon_\refer\right) t_{\max}}
	\right\rbrace \\
	&\cdot	\int_{\svar_j}^{t} e^{-\epsilon_\refer (t-\tau)} \alpha_w\left(\abs{w(\tau)}\right) d\tau\\
	\leq&  e^{ -\epsilon_\ivar{j} (t-\svar_j)}V(x(\tvar_j)) +\maxvar  \int_{\svar_j}^{t}e^{-\epsilon_\refer (t-\tau)} \alpha_w\left(\abs{w(\tau)}\right) d\tau,
	\end{split}
	\end{equation}
	where $U_\ivar{j}(\xi)$ is the function according to \eqref{eq_def_u} for the parameters $\gamma_\ivar{j}$ and $\Lambda_\ivar{j}$, holds for	 
	\begin{equation}
	\label{eq_def_maxvar}
	\begin{split}
	&\maxvar 
	\coloneqq \underset{i \in \left\lbrace 1,\dots,n_\tpar\right\rbrace}{\max} \left\lbrace \max \left\lbrace e^{\epsilon_\refer t_{\max}},e^{\left(-\epsilon_i+\epsilon_\refer\right) t_{\max}}
	\right\rbrace\right\rbrace.
	\end{split}			
	\end{equation}
	
	 Next, we show by induction that
	 \begin{equation}
	 \label{eq_fir_iss_res}
	 \begin{split}
	 V(x(t,j))
	 \leq& e^{-\bar\epsilon t}  \max \left\lbrace V(x(0,0)), \abs{\eta(0,0)} \right\rbrace\\
	 &+ \maxvar \int_{0}^{t} e^{-\bar \epsilon(t-\tau)} \alpha_w(\abs{w(\tau)}) d\tau		
	 \end{split}
	 \end{equation}
	 holds with $\bar \epsilon = \min \left\lbrace \epsilon_1,\epsilon_\refer \right\rbrace$ for all $(t,j) \in \text{dom}~\xi$. It trivially holds for $t = j = 0$. Further, suppose it holds  for all $\tvar_j$ with $j \leq \barj$ for some $\barj \in \mathbb{N}$. 
	 
	 Plugging \eqref{eq_fir_iss_res} for $(t,j) = \tvar_{\barj-m+k}$ into \eqref{eq_fir_eta}, we obtain for $m-1 \geq k \geq m-\barj$ that	 
	 \begin{equation}
	 \label{eq_eta_fir_res}
	 \begin{split}
	 &\eta_k(s_\barj)\\ 
	 \leq& e^{-\bar\epsilon \svar_\barj}  \max \left\lbrace V(x(0,0)), \abs{\eta(0,0)} \right\rbrace +  e^{-\epsilon_\refer(\svar_\barj-\svar_{\barj-m+k})}\maxvar\\
	 &\cdot	 \int_{0}^{\svar_{\barj-m+k}} e^{-\bar\epsilon(\svar_{\barj-m+k}-\tau)} \alpha_w\left(\abs{w(\tau)}\right) d\tau\\
	 \leq &  e^{-\bar\epsilon \svar_{\barj}} \max \left\lbrace V(x(0,0)), \abs{\eta(0,0)} \right\rbrace\\
	 &+\maxvar\int_{0}^{\svar_{\barj-1}} e^{-\bar\epsilon(\svar_{\barj}-\tau)} \alpha_w\left(\abs{w(\tau)}\right) d\tau.
	 \end{split} 
	 \end{equation}
	 For $k < m-\barj$, we obtain due to the update of $\eta$ according to \eqref{eq_S_window} that
	 \begin{equation}
	 \label{eq_eta_fir_res2}
	 	\eta_k(\tvar_{\barj}) = e^{-\epsilon_\refer\svar_{\barj}} \eta_{k+{\barj}}(0,0)
	 \end{equation}
	 holds.  Combining \eqref{eq_fir_iss_res}, \eqref{eq_eta_fir_res} and \eqref{eq_eta_fir_res2}, we can conclude that
	 \begin{equation}
		 \label{eq_eta_fir_res3}
		 \begin{split}
			 &\frac{1}{m} \left(V(x(\tvar_{\barj}))+\sum_{k=1}^{m-1} \eta_k(\tvar_{\barj})\right)\\
			 \leq& e^{-\bar\epsilon \svar_\barj} \max \left\lbrace V(x(0,0)), \abs{\eta(0,0)} \right\rbrace\\
			  &+ \maxvar\int_{0}^{\svar_\barj} e^{-\bar\epsilon(\svar_\barj-\tau)} \alpha_w\left(\abs{w(\tau)}\right) d\tau .
		 \end{split}		
	 \end{equation} 
	 Recall from Subsection~\ref{sec_def_algo} that Algorithm~\ref{algo_trig_window} ensures for $C(x,\eta)$ according to \eqref{eq_fir_C} either that 
 	\begin{equation}
 	\label{eq_tr_fir_1}
 	\begin{split}	
	 	&e^{ -\epsilon_\ivar{\barj}  \left(t-\svar_\barj\right) }V(x(\tvar_j))\\
	 	 	\leq& e^{-\epsilon_\refer \left(t-\svar_\barj\right)}  \frac{1}{m}\left(V(x(\tvar_j))+\sum_{k=1}^{m-1} \eta_k(\tvar_j)\right)\\
	 	\overset{\eqref{eq_eta_fir_res3}}{\leq} & e^{-\bar\epsilon t} \max \left\lbrace V(x(0,0)), \abs{\eta(0,0)} \right\rbrace\\
	 	&+ \maxvar\int_{0}^{\svar_\barj} e^{-\bar\epsilon(t-\tau)} \alpha_w\left(\abs{w(\tau)}\right) d\tau 
	 	\end{split}
 	\end{equation}
 	holds for $\svar_\barj \leq t \leq \svar_{\barj+1}$ for some $i_\barj$ with $\svar_{j+1} - \svar_j \leq \delta T_{\max}\left(\gamma_\ivar{j},L_\ivar{j}+\frac{\epsilon_\ivar{j}}{2}\right)$ or, if it uses the fall-back strategy, i.e., if $\ivar{\barj} = 1$, that
 	\begin{equation}
	 	\begin{split}	
		   &\hphantom{\leq}~\,e^{-\epsilon_1 \left(t-\svar_\barj\right)} V(x(\tvar_j))\\
		 	\overset{\eqref{eq_fir_iss_res} \text{ for } \tvar_\barj}&{\leq} e^{-\bar \epsilon t}  \max \left\lbrace V(x(0,0)), \abs{\eta(0,0)} \right\rbrace\\
		 	&\phantom{\leq}+\maxvar \int_{0}^{\svar_\barj} e^{-\bar \epsilon(t-\tau)}\alpha_w(\abs{w(\tau)}) d\tau		
	 	\end{split}	
	 	\label{eq_tr_fir_2}
 	\end{equation}
 	holds for $\svar_\barj \leq t \leq \svar_{\barj+1}$. Using \eqref{eq_tr_fir_1} and \eqref{eq_tr_fir_2} in \eqref{eq_V_dec_fir}, it follows that 
 	\begin{equation*}
 		\begin{split}
	 		V(x(\svar,\barj+1))~~&\\
	 		 \leq U_\ivar{\barj}(\xi(t,\barj+1))& \leq e^{-\bar\epsilon t} \max \left\lbrace V(x(0,0)), \abs{\eta(0,0)}\right\rbrace\\
	 		 &\phantom{\leq}+\maxvar \int_{0}^{t} e^{-\bar \epsilon(t-\tau)}\alpha_w(\abs{w(\tau)}) d\tau 
 		\end{split}
 	\end{equation*} 	
 	holds for $\svar_\barj \leq t \leq \svar_{\barj+1}$ and thus also for $\tvar_{\barj+1}$. 
 	Thus \eqref{eq_fir_iss_res} holds by induction for all $(t,j)\in \text{dom}~\xi$.
 	
 	Further, since $\phi \in \left[\lambda,\lambda^{-1}\right]$ for some $\lambda\in \left(0,1\right)$ in the definition of $U_{i_j}$, there exists a function $\phi_1$ such that for all $\xi$ and all ${i_j} \in \left\lbrace 1,\dots,n_\tpar\right\rbrace$, 
 		$\phi_1\left(\abs{\begin{bmatrix}
 			x\\
 			e
 			\end{bmatrix}}\right) \leq U_{i_j}(\xi) 
		$ holds, which implies
 		\begin{equation*}
	 		\begin{split}
	 			\abs{\begin{bmatrix}
	 				x(t,j)\\
	 				e(t,j)
 				\end{bmatrix}} \leq& \phi_1^{-1}\left( U_{\ivar{j}}(\xi(t,j))\right)\\
 			\leq& \phi_1^{-1}\left( 2 e^{-\bar\epsilon t} \max \left\lbrace V(x(0,0)), \abs{\eta(0,0)}\right\rbrace\right)\\
 			&+  \phi_1^{-1}\left(2 \maxvar \int_{0}^{t} e^{-\bar \epsilon(t-\tau)}\alpha_w(\abs{w(\tau)}) d\tau\right)\\
 			\leq&  \phi_1^{-1}\left( 2 e^{-\bar\epsilon \left(\frac{t+j t_{\min}}{2}\right)} \max \left\lbrace V(x(0,0)), \abs{\eta(0,0)}\right\rbrace\right)\\
 			 &+ \phi_1^{-1}\left(\frac{2\maxvar}{\bar\epsilon}\alpha_w\left(\norm{w}_\infty\right) \right).
	 		\end{split}
 		\end{equation*}
 		Here, we used that since $t_{j+1} - t_j \geq t_{\min}$, $t \geq \frac{t+j t_{\min}}{2}$ holds for all $(t,j) \in \text{dom}~\xi$.
 		
 		Moreover, it holds for $\svar_{j-1} \leq t \leq \svar_j$ due to the update of $\eta$ at sampling instants and due to \eqref{eq_eta_fir_res} and \eqref{eq_eta_fir_res2} that
 		\begin{equation}
	 		\begin{split}
		 		&\abs{\eta(t,j)} = \sum_{k = 1}^{m-1}\abs{\eta_k(\tvar_j)}\\
		 		 \leq& m  e^{-\bar\epsilon \svar_{\barj}} \max \left\lbrace V(x(0,0)), \abs{\eta(0,0)} \right\rbrace\\
		 		&+m\maxvar\int_{0}^{\svar_{\barj-1}} e^{-\bar\epsilon(\svar_{\barj}-\tau)} \alpha_w\left(\abs{w(\tau)}\right) d\tau\\
		 		\leq & m  e^{-\bar\epsilon \left(\frac{t+j t_{\min}}{2}\right)} \max \left\lbrace V(x(0,0)), \abs{\eta(0,0)} \right\rbrace\\
		 		&+\frac{m\maxvar}{\bar\epsilon}\alpha_w\left(\norm{w}_\infty\right).
	 		\end{split}			
 		\end{equation}
 	 We thus obtain
 		\begin{equation*}
 		\begin{split}
	 		&\abs{\begin{bmatrix}
 			x(t,j)\\
 			e(t,j)\\
 			\eta(t,j)
 			\end{bmatrix}} \leq \abs{\begin{bmatrix}
 			x(t,j)\\
 			e(t,j)
 			\end{bmatrix}} + \abs{\begin{bmatrix}
 			\eta(t,j)
 			\end{bmatrix}}\\
 		 \leq& \phi_1^{-1}\left( 2 e^{-\bar\epsilon \left(\frac{t+j t_{\min}}{2}\right)} \max \left\lbrace V(x(0,0)), \abs{\eta(0,0)}\right\rbrace\right)\\ 
 		 &+me^{-\bar\epsilon \left(\frac{t+j t_{\min}}{2}\right)} \max \left\lbrace V(x(0,0)), \abs{\eta(0,0)} \right\rbrace\\
 		 &+ \phi_1^{-1}\left(\frac{2\maxvar}{\bar\epsilon}\alpha_w\left(\norm{w}_\infty\right) \right) + \frac{m\maxvar}{\bar\epsilon}\alpha_w\left(\norm{w}_\infty\right),
 		\end{split}
 		\end{equation*}
 		 which proves ISS of $\mathcal{H}_{STC}$ trivially. \hfill\hfill\qed
\subsubsection{Proof of Theorem~\ref{theo_iir}}
\label{proof_theo_iir}
This proof follows the same lines as the proof of Theorem~\ref{prop_fir}. We thus sketch here only the differences. Similar as in the proof of Theorem~\ref{prop_fir}, we obtain that 		
	$t_{\min} \leq \svar_{j+1} - \svar_j \leq t_{\max} \coloneqq \underset{i \in \left\lbrace 1,\dots,n_\tpar\right\rbrace}{\max} \delta T_{\max}(\gamma_i,L_i+\frac{\epsilon_i}{2})
	$ holds and that there is for each $j \in \mathbb{N}_0$ an $\ivar{j} \in \left\lbrace 1,\dots,n_\tpar\right\rbrace$, such that $\svar_{j+1} - \svar_j \leq \delta T_{\max}\left(\gamma_\ivar{j},L_\ivar{j}+\frac{\epsilon_\ivar{j}}{2}\right)$.  Proposition~\ref{prop_hybrid} thus implies
	for $\svar_j \leq t \leq \svar_{j+1}$ with $\Lambda_\ivar{j} = L_\ivar{j}+\frac{\epsilon_\ivar{j}}{2}$ that 
	\begin{equation}
	\label{eq_V_dec_iir}
	\begin{split}
	&V(x(\svar,j+1))\leq U_\ivar{j}(\xi(t,j+1))	\\
	\leq&  e^{ -\epsilon_\ivar{j} (t-\svar_j)}V(x(\tvar_j)) +\maxvar  \int_{\svar_j}^{t}e^{-\epsilon_\refer (t-\tau)} \alpha_w\left(\abs{w(\tau)}\right) d\tau,
	\end{split}
	\end{equation}
	where $U_\ivar{j}(\xi)$ is the function according to \eqref{eq_def_u} for the parameters $\gamma_\ivar{j}$ and $\Lambda_\ivar{j}$, holds for $\maxvar$ according to \eqref{eq_def_maxvar}. 
	
	The next step is, similar as in the proof of Theorem~\ref{theo_iir}, to show by induction that, for all $(t,j) \in \text{dom}~\xi$,
	\begin{equation}
	\label{eq_iir_iss_res}
	\begin{split}
	V(x(t,j))
	\leq& e^{-\bar\epsilon t}  \max \left\lbrace V(x(0,0)), \abs{\eta(0,0)} \right\rbrace\\
	&+ \maxvar \int_{0}^{t} e^{-\bar \epsilon(t-\tau)} \alpha_w(\abs{w(\tau)}) d\tau		
	\end{split}
	\end{equation}
	and
	\begin{equation}
		\label{eq_iir_ind_eta}
		\begin{split}
			\eta(t,j) \leq& e^{-\bar\epsilon t} \max \left\lbrace V(x(0,0)), \abs{\eta(0,0)} \right\rbrace\\
			&+ \maxvar \int_{0}^{t} e^{-\bar \epsilon(t-\tau)} \alpha_w(\abs{w(\tau)}) d\tau
		\end{split}		
	\end{equation}
	hold with $\bar \epsilon = \min \left\lbrace \epsilon_1,\epsilon_\refer \right\rbrace$ . Both inequalities trivially hold for $t = j = 0$. Further, suppose the inequalities hold for all $\tvar_j$  with $j \leq \barj$ for some $\barj \in \mathbb{N}$. 		
	Plugging \eqref{eq_iir_iss_res} and \eqref{eq_iir_ind_eta} for $(t,j) = \tvar_{\barj}$ into \eqref{eq_S_iir}, we obtain for $\svar_\barj \leq t \leq \svar_{\barj+1}$ since $r_1+r_2 \leq 1$ that	 
	\begin{equation}
	\label{eq_eta_iir_res}
	\begin{split}
	\eta(t,\barj+1) =& S(\eta(\tvar_\barj),x(\tvar_\barj))\\
	\leq& e^{-\bar \epsilon \svar_{\barj+1}} \max \left\lbrace V(x(0,0)), \abs{\eta(0,0)} \right\rbrace\\
	&+ \maxvar \int_{0}^{t} e^{-\bar \epsilon(\svar_{\barj+1}-\tau)} \alpha_w(\abs{w(\tau)}) d\tau.	
	\end{split} 
	\end{equation}
	Recall from Subsection~\ref{sec_def_algo} that Algorithm~\ref{algo_trig_window} ensures for $C(x,\eta)$ according to \eqref{eq_iir_C} either that 
	\begin{equation}
	\label{eq_tr_iir_1}
	\begin{split}	
	&e^{ -\epsilon_\ivar{\barj}  \left(t-\svar_\barj\right) }V(x(\tvar_j))	\leq e^{-\epsilon_\refer \left(t-\svar_\barj\right)}  \eta(\tvar_j)\\
	\overset{\eqref{eq_iir_ind_eta} \text{ for } \tvar_\barj}{\leq} & e^{-\bar\epsilon t} \max \left\lbrace V(x(0,0)), \abs{\eta(0,0)} \right\rbrace\\
	&+ \maxvar\int_{0}^{\svar_\barj} e^{-\bar\epsilon(t-\tau)} \alpha_w\left(\abs{w(\tau)}\right) d\tau 
	\end{split}
	\end{equation}
	or, if it uses the fall-back strategy, i.e., if $\ivar{\barj} = 1$, that
	\begin{equation}
	\begin{split}	
	&\hphantom{\leq}~\,e^{-\epsilon_1 \left(t-\svar_\barj\right)} V(x(\tvar_\barj))\\
	\overset{\eqref{eq_iir_iss_res} \text{ for } \tvar_\barj}&{\leq} e^{-\bar \epsilon t}  \max \left\lbrace V(x(0,0)), \abs{\eta(0,0)} \right\rbrace\\
	&\phantom{\leq}+\maxvar \int_{0}^{\svar_\barj} e^{-\bar \epsilon(t-\tau)}\alpha_w(\abs{w(\tau)}) d\tau		
	\end{split}	
	\label{eq_tr_iir_2}
	\end{equation}
	hold for $\svar_\barj \leq t \leq \svar_{\barj+1}$. Using \eqref{eq_tr_iir_1} and \eqref{eq_tr_iir_2} in \eqref{eq_V_dec_iir}, it follows that 
	\begin{equation*}
	\begin{split}
		V(x(\svar,\barj+1))~&\leq\\
	U_\ivar{\barj}(\xi(t,\barj+1))& \leq e^{-\bar\epsilon t} \max \left\lbrace V(x(0,0)), \abs{\eta(0,0)}\right\rbrace\\
	&\phantom{\leq}+\maxvar \int_{0}^{t} e^{-\bar \epsilon(t-\tau)}\alpha_w(\abs{w(\tau)}) d\tau 
	\end{split}
	\end{equation*}	
	holds for $\svar_\barj \leq t \leq \svar_{\barj+1}$ and thus also for $\tvar_{\barj+1}$. 
	Thus \eqref{eq_iir_iss_res} and \eqref{eq_iir_ind_eta} hold by induction for all $(t,j)\in\text{dom}~\xi$. The remainder of this proof is similar to the corresponding part of the proof of Theorem~\ref{prop_fir} and thus omitted. \hfill\hfill\qed
\subsubsection{Proof of Theorem~\ref{theo_ref}}
\label{proof_theo_ref}
	This proof follows the same lines as the proof of Theorem~\ref{prop_fir}. We thus sketch here only the differences. Similar as in the proof of Theorem~\ref{prop_fir}, we obtain that 		
	$t_{\min} \leq \svar_{j+1} - \svar_j \leq t_{\max} \coloneqq \underset{i \in \left\lbrace 1,\dots,n_\tpar\right\rbrace}{\max} \delta T_{\max}(\gamma_i,L_i+\frac{\epsilon_i}{2})
	$ holds and that there is for each $j \in \mathbb{N}_0$ an $\ivar{j} \in \left\lbrace 1,\dots,n_\tpar\right\rbrace$, such that $\svar_{j+1} - \svar_j \leq \delta T_{\max}\left(\gamma_\ivar{j},L_\ivar{j}+\frac{\epsilon_\ivar{j}}{2}\right)$.  Proposition~\ref{prop_hybrid} thus implies
	for $\svar_j \leq t \leq \svar_{j+1}$ with $\Lambda_\ivar{j} = L_\ivar{j}+\frac{\epsilon_\ivar{j}}{2}$ that 
	\begin{equation}
	\label{eq_V_dec_ref}
	\begin{split}
	&V(x(\svar,j+1)) \leq U_\ivar{j}(\xi(\svar,j+1))	\\
	\leq&  e^{ -\epsilon_\ivar{j} (t-\svar_j)}V(x(\tvar_j)) +\maxvar  \int_{\svar_j}^{t}e^{-\epsilon_\refer (t-\tau)} \alpha_w\left(\abs{w(\tau)}\right) d\tau,
	\end{split}
	\end{equation}
	where $U_\ivar{j}(\xi)$ is the function according to \eqref{eq_def_u} for the parameters $\gamma_\ivar{j}$ and $\Lambda_\ivar{j}$, holds for $\maxvar$ according to \eqref{eq_def_maxvar}.

	The next step is, similar as in the proof of Theorem~\ref{theo_iir}, to show by induction that, for all $(t,j) \in \text{dom}~\xi$,
	\begin{equation}
	\label{eq_ref_iss_res}
	\begin{split}
	V(x(t,j))
	\leq& e^{-\bar\epsilon t} V(x(0,0)) \\
	&+ \maxvar \int_{0}^{t} e^{-\bar \epsilon(t-\tau)} \alpha_w(\abs{w(\tau)}) d\tau		
	\end{split}
	\end{equation}
	holds with $\bar \epsilon = \min \left\lbrace \epsilon_1,\epsilon_\refer \right\rbrace$ . It trivially holds for $t = j = 0$. Further suppose it holds  for all $\tvar_j$ with $j \leq \barj$ for some $\barj \in \mathbb{N}$. 
	Recall from Subsection~\ref{sec_def_algo} that Algorithm~\ref{algo_trig_window} ensures for $C(x,\eta)$ according to \eqref{eq_ref_C} either that 
	\begin{equation}
	\label{eq_tr_ref_1}
	\begin{split}	
	e^{ -\epsilon_\ivar{\barj}  \left(t-\svar_\barj\right) }V(x(\tvar_\barj))	\leq e^{-\epsilon_\refer \left(t-\svar_\barj\right)}  \eta(\tvar_\barj)
	\leq e^{-\bar \epsilon t} V(x(0,0)) 
	\end{split}
	\end{equation}
	or, if it uses the fall-back strategy, i.e., if $\ivar{\barj} = 1$, that
	\begin{equation}
	\begin{split}	
	&\hphantom{\leq}~\,e^{-\epsilon_1 \left(t-\svar_\barj\right)} V(x(\tvar_\barj))\\
	\overset{\eqref{eq_ref_iss_res} \text{ for } \tvar_\barj}&{\leq} e^{-\bar \epsilon t}  \max \left\lbrace V(x(0,0)), \abs{\eta(0,0)} \right\rbrace\\
	&\phantom{\leq}+\maxvar \int_{0}^{\svar_\barj} e^{-\bar \epsilon(t-\tau)}\alpha_w(\abs{w(\tau)}) d\tau		
	\end{split}	
	\label{eq_tr_ref_2}
	\end{equation}
	holds for $\svar_\barj \leq t \leq \svar_{\barj+1}$. Using \eqref{eq_tr_ref_1} and \eqref{eq_tr_ref_2} in \eqref{eq_V_dec_ref}, it follows that 
	\begin{equation*}
	\begin{split}
		V(x(\svar,\barj+1))~&\leq\\
	U_\ivar{\barj}(\xi(t,\barj+1)) &\leq e^{-\bar\epsilon t} \max \left\lbrace V(x(0,0)), \abs{\eta(0,0)}\right\rbrace\\
	&\phantom{\leq}+\maxvar \int_{0}^{t} e^{-\bar \epsilon(t-\tau)}\alpha_w(\abs{w(\tau)}) d\tau 
	\end{split}
	\end{equation*}	
	holds for $\svar_\barj \leq t \leq \svar_{\barj+1}$ and thus also for $\tvar_{\barj+1}$. 
	Thus \eqref{eq_ref_iss_res} holds by induction for all $(t,j)\in\text{dom}~\xi$. The remainder of this proof is similar to the corresponding part of the proof of Theorem~\ref{prop_fir} and thus omitted. \hfill\hfill\qed
\subsubsection{Proof of Proposition~\ref{prop_hybrid_loc}}
\label{proof_prop_hyb_loc}
\begin{proof}
	Similar as in the proof of Proposition~\ref{prop_hybrid}, we obtain that 
	\begin{equation*}
	\label{eq_dec_u_loc}
	\begin{split}
	&\frac{d}{dt} U(\xi(\svar,j+1))
	\leq  -\epsilon  U(\xi(t,j+1))+\alpha_w\left(\abs{w}\right) 
	\end{split}	
	\end{equation*} 
	holds if $x(t,j+1) \in \mathcal{X}_c$ and $\hat{x}(t,j+1) = \hat{x}(s_{j}^+) = x(s_j) \in \mathcal{X}_c$. Note that $V(x(\tvar_j)) < c$ and  $e(\tvar_j) = 0$ such that $x(\tvar_j) = \hat{x}(\tvar_j)\in \mathcal{X}$. 
	Thus, we obtain for $t$ sufficiently close to $\svar_j$ due to the comparison Lemma (cf. \cite[p. 102]{khalil2002nonlinear}) that
	\begin{align*}
	& U(\xi(\svar,j+1)) \nonumber\\ 
	\leq& e^{ -\epsilon (t-\svar_j)}V(x(\tvar_j))\nonumber 
	+ \int_{t_j}^{t} e^{\ -\epsilon (t-\tau)} \alpha_w\left(\abs{w(\tau)}\right) d\tau \nonumber\\
	\overset{\abs{w(\tau)}\leq\bar w}{\leq}&  e^{-\epsilon (t-\svar_j)}V(x(\tvar_j)) + \frac{\alpha_w\left(\abs{\bar w}\right)}{\epsilon} \left(1-e^{-\epsilon\left(t-\svar_j\right)}\right). \nonumber
	\end{align*}
	From \eqref{eq_bound_U_loc} and $V(x(\tvar_j)) \leq c$, we obtain using simple derivative arguments that $e^{-\epsilon (t-\svar_j)}V(x(\tvar_j)) + \frac{\alpha_w\left(\abs{\bar w}\right)}{\epsilon} \left(1-e^{-\epsilon\left(t-\svar_j\right)}\right) \leq c$ and thus that $x(t) \in \mathcal{X}$ for $t$ sufficiently close to $\svar_j$. 
	We can now use this argumentation iteratively to observe with $V(x(\svar,j+1)) \leq U(\xi(\svar,j+1))$ that \eqref{eq_prop_hybrid1_loc} holds for $\svar_j \leq t \leq \svar_j + \auxvar(\svar_j^+)$.
\end{proof}				
\subsubsection{Proof of Theorem~\ref{prop_fir_loc}}
\label{proof_prop_fir_loc}
			Similar as in the proof of Theorem~\ref{prop_fir}, we obtain that 		
			$t_{\min} \leq \svar_{j+1} - \svar_j \leq t_{\max} \coloneqq \underset{i \in \left\lbrace1,\dots,n_\tpar\right\rbrace, l \in\left\lbrace1,\dots,n_c \right\rbrace}{\max} \delta T_{\max}\left(\gamma_\newi,L_\newi+\frac{\epsilon_\newi}{2}\right)
			$ holds.

			Now, we will show by induction that 
			\begin{equation}
				\label{eq_fir_loc_ind}
				\begin{split}
					&V(x(t,j)) \\
					\leq& \min \left\lbrace c_w + e^{-\bar\epsilon t} \max \left\lbrace V(x(0,0))-c_w,\abs{\eta(0,0)} \right\rbrace, c_{\max} \right\rbrace.
				\end{split}
			\end{equation}
			holds for all $(t,j) \in \text{dom}~\xi$ and $\bar \epsilon = \min\left\lbrace \underset{l \in \left\lbrace 1,\dots,n_c \right\rbrace}{\min} \epsilon_{1,l},\epsilon_\refer\right\rbrace$. It trivially holds for $t = j = 0$ if $x(\tvar_0) \in \mathcal{X}_{c_{\max}}$. Now suppose it holds for all $\tvar_j$ with $j \leq \barj$ for some $\barj \in \mathbb{N}$. 
			Let $l_\barj =  \underset{l\in\left\lbrace 1,\dots,n_c \right\rbrace}{\min}~l$ s.t. $c_l \geq \max\left\lbrace V(x(\tvar_{\barj})),C(x(\tvar_{\barj}),\eta(\tvar_{\barj})) \right\rbrace$, which is selected in Line~\ref{line_l} in Algorithm~\ref{algo_trig_window_loc}. 		
			 We can conclude from the algorithm for $\tilde{j}$ and $l_\barj$ that either \eqref{eq_dec_gen_loc} and \eqref{eq_bound_aux_loc} hold for some $\ivar{\barj}\in\left\lbrace 2,\dots,n_\tpar\right\rbrace$ and $\svar_\barj \leq t \leq \svar_{\barj+1}$, or that $\svar_{\barj+1} = \svar_{\barj} + T_{\max}(\gamma_{1,l_\barj},L_{1,l_\barj}+\frac{\epsilon_{1,l_\barj}}{2})$.
			 
			 In the former case, \eqref{eq_dec_gen_loc} implies since $\max \left\lbrace V(x(\tvar_j)),C(x(\tvar_j),\eta(\tvar_j))\right\rbrace \leq c_{l_\barj}$ that \eqref{eq_bound_U_loc} holds for $\barj$ with $\epsilon = \epsilon_{i_\barj,l_\barj}$. Thus, it follows in this case from Proposition~\ref{prop_hybrid_loc} 
			 that \eqref{eq_cond_dec} holds for $\barj$ and $\svar_\barj \leq t \leq \svar_{\barj+1}$. Similar as in the proof of Theorem~\ref{prop_fir}, we obtain due to the update of $\eta$ according to \eqref{eq_S_window} with \eqref{eq_fir_loc_ind} for $(t,j) = \tvar_j$ with $j \leq \barj$ that 
			 \begin{equation}
				 \label{eq_fir_loc_C}
				 \begin{split}
				    &C(x(\tvar_\barj),\eta(\tvar_\barj))\\
	 			 	\leq&\frac{1}{m} \left(V(x(\tvar_{\barj})) + \sum_{k=1}^{m-1} \eta_k(\tvar_\barj)\right)\\
	 			 	\leq& \min \left\lbrace c_w + e^{-\bar\epsilon \svar_\barj} \max \left\lbrace V(x(0,0))-c_w,\abs{\eta(0,0)} \right\rbrace, c_{\max} \right\rbrace
				 \end{split}
			 \end{equation}
			 Plugging \eqref{eq_fir_loc_C} in \eqref{eq_cond_dec}, we can conclude that \eqref{eq_fir_loc_ind} holds in this case for $\barj+1$ and $\svar_\barj \leq t \leq \svar_{\barj+1}$.
			 
			 In the other case, i.e., if $\svar_{\barj+1} = \svar_{\barj} + T_{\max}\left(\gamma_{1,l_\barj},L_{1,l_\barj}+\frac{\epsilon_{1,l_\barj}}{2}\right)$, we observe that \eqref{eq_bound_U_loc} holds for $\epsilon = \epsilon_{1,l_\barj}$ since $V(x(\tvar_\barj)) \leq c_{l_\barj}$, $\frac{\alpha_w\left(\abs{\bar w}\right)}{\epsilon_{1,l_\barj}} \leq c_w \leq c_{l_\barj}$ and $e^{-\epsilon_1 \auxvar(\tvar_\barj^+)} \leq 1$. Thus, we can use in this case Proposition~\ref{prop_hybrid_loc} and obtain that 
			 \begin{equation}
			 \label{eq_fir_eps1_loc}
			 	\begin{split}
			 	V(x(t,\barj+1)) 
			 	\leq&	e^{ -\epsilon_{1,l_\barj} (t-\svar_\barj)}V(x(\tvar_\barj))\\
			 	& + \frac{\alpha_w\left(\abs{\bar w}\right)}{\epsilon_{1,l_\barj}} \left(1-e^{-\epsilon_{1,l_\barj}(t-\svar_\barj)}\right)
			 	\end{split}
			 \end{equation} 
			 holds for $\svar_\barj \leq t \leq \svar_{\barj+1}$. Plugging \eqref{eq_fir_loc_ind} for $(t,j) = \tvar_\barj$ into \eqref{eq_fir_eps1_loc}, we obtain for $\svar_\barj\leq t \leq \svar_{\barj+1}$
			 \begin{equation}
			 	\begin{split}
			 	&V(x(t,\barj+1))\\ 
			 	\leq& \frac{\alpha_w\left(\abs{\bar w}\right)}{\epsilon_{1,l_\barj}} \left(1-e^{-\epsilon_{1,l_\barj}(t-\svar_\barj)}\right) +	e^{ -\epsilon_{1,l_\barj} (t-\svar_\barj)}\\
			 	&\cdot\min \left\lbrace c_w + e^{-\bar\epsilon \svar_\barj} \max \left\lbrace V(x(0,0))-c_w,\abs{\eta(0,0)} \right\rbrace, c_{\max} \right\rbrace\\	
			 	\leq& c_w \left(1-e^{-\epsilon_{1,l_\barj}(t-\svar_\barj)}\right)+ e^{ -\epsilon_{1,l_\barj} (t-\svar_\barj)}\\
			 	&\cdot \min \left\lbrace c_w + e^{-\bar\epsilon \svar_\barj} \max \left\lbrace V(x(0,0))-c_w,\abs{\eta(0,0)} \right\rbrace, c_{\max} \right\rbrace\\	 	
			 	\leq& \min \left\lbrace c_w + e^{-\bar\epsilon t} \max \left\lbrace V(x(0,0))-c_w,\abs{\eta(0,0)} \right\rbrace, c_{\max} \right\rbrace. 
			 	\end{split}
			 \end{equation}
			 Thus \eqref{eq_fir_loc_ind} holds in both cases  also for $\tvar_{\barj+1}$ and hence by induction for all $(t,j) \in\text{dom}~\xi$. RAS of $\mathcal{R}$ follows immediately from \eqref{eq_fir_loc_ind} and $t \geq \frac{t+j t_{\min}}{2}$, which implies that we can chose
			 \begin{equation*}
				 \begin{split}
					 &\beta\left(\max\left\lbrace V(x(0,0)) - c_w, 0 \right\rbrace + \abs{\begin{bmatrix}
					 		e(0,0)\\
					 		\eta(0,0)
					 \end{bmatrix}},t, j\right)\\
					  =& 2e^{-\bar \epsilon \left(\frac{t+jt_{\min}}{2}\right)} \left(\max \left\lbrace V(x(0,0))-c_w,0 \right\rbrace + \abs{\eta(0,0)}\right).\hfill\hfill\qed
				 \end{split}			 	
			 \end{equation*}

\end{document}